\RequirePackage{fix-cm}
\documentclass{article}
\usepackage{pacf_preprint,times}

\usepackage[utf8]{inputenc}
\usepackage[T1]{fontenc}
\usepackage{hyperref}
\usepackage{url}
\usepackage{booktabs}
\usepackage{amsfonts}
\usepackage{nicefrac}
\usepackage{microtype}
\usepackage{xcolor}
\usepackage{multirow}
\usepackage{subcaption}
\usepackage{enumitem}
\usepackage{graphicx}
\usepackage{wrapfig}
\usepackage{amsmath,amssymb,amsthm,mathtools,bm}
\usepackage{algorithm}
\usepackage{algpseudocode}

\DeclareMathOperator{\GELU}{GELU}
\DeclareMathOperator{\LN}{LN}
\DeclareMathOperator{\Unif}{Unif}
\DeclareMathOperator{\Beta}{Beta}
\DeclareMathOperator{\supp}{supp}
\DeclareMathOperator{\tr}{tr}
\DeclareMathOperator{\RKfour}{RK4}

\newtheorem{theorem}{Theorem}
\newtheorem{proposition}[theorem]{Proposition}

\theoremstyle{remark}
\newtheorem{remark}[theorem]{Remark}

\title{Pathwise Alignment for \\
Collaborative Filtering: \\
Bilateral Conditional Flow Regularization}

\author{\parbox{\dimexpr\textwidth-2\tabcolsep\relax}{\centering\normalfont
Zhiwei Li$^{1}$, Guodong Long$^{1}$, Chengqi Zhang$^{2}$, and Jing Jiang$^{1}$\\[4pt]
\small $^{1}$Australian Artificial Intelligence Institute\\
Faculty of Engineering and Information Technology\\
University of Technology Sydney, Sydney, NSW, Australia\\[3pt]
$^{2}$Department of Data Science and Artificial Intelligence\\
The Hong Kong Polytechnic University, Hong Kong SAR, China\\[3pt]
\texttt{zhw.li@outlook.com}; \texttt{guodong.long@uts.edu.au}\\
\texttt{chengqi.zhang@polyu.edu.hk}; \texttt{jing.jiang@uts.edu.au}
}}
\arxivpreprintcopy

\hypersetup{hidelinks,
pdftitle={Pathwise Alignment for Collaborative Filtering: Bilateral Conditional Flow Regularization},
pdfauthor={Zhiwei Li, Guodong Long, Chengqi Zhang, Jing Jiang}}
\begin{document}
\maketitle

\begin{abstract}
Collaborative filtering predicts user-item interactions from implicit feedback.
We focus on matrix factorization (MF) trained with the Bayesian Personalized Ranking loss.
This loss updates each user embedding only through that user's observed items and sampled negatives, so sparse feedback gives MF few constraints.
To constrain each user through other users, we propose \textbf{P}athwise \textbf{A}lignment for \textbf{C}ollaborative \textbf{F}iltering (PACF), which trains MF with one network shared by all training pairs.
At points between user and item embeddings, a flow matching term trains the network to predict the user-to-item step, reversed and scaled for sampled negatives.
A trajectory consistency term pulls each user's predicted path to end at the observed item.
Both terms update the embeddings, so the network is discarded after training, thereby maintaining the efficiency of MF at inference.
We prove that a network with bounded sensitivity to its input point accurately fits users who share an item only if their embeddings are close, so these users constrain one another.
Against eleven discriminative and generative baselines, PACF achieves the highest performance on all seven datasets.
Its relative gains over MF are largest where feedback is sparsest.
The shared network thus supplies the missing constraints through other users.
\end{abstract}

\section{Introduction}

\begin{wrapfigure}[18]{r}{0.5\columnwidth}
  \centering
  \includegraphics[width=0.9\linewidth]{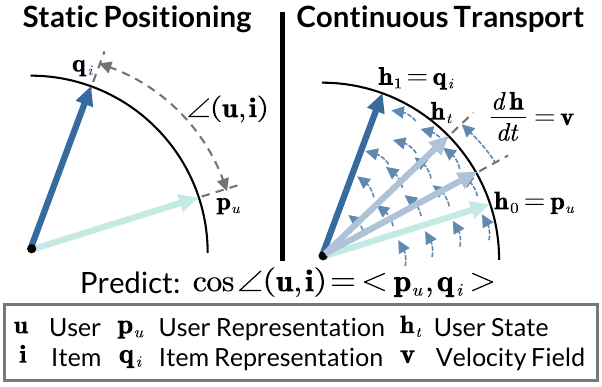}
  \caption{
Comparison of supervision in unit space.
The left panel shows scalar endpoint supervision.
The right panel shows PACF fitting signed displacements at sampled states along user-item segments. 
PACF thus supervises the whole path from each user to its item instead of one endpoint score.
  }
  \label{fig:sc}
\end{wrapfigure}

Collaborative Filtering (CF) learns user and item embeddings from sparse interactions and scores pairs by their inner product~\citep{goldberg1992using, papadakis2022collaborative, he2017neural}.
Pairwise ranking losses supervise these scores, and inference uses the learned embeddings for static similarity scoring~\citep{hsieh2017collaborative}.

A pairwise loss constrains the inner products of observed and sampled pairs, leaving the remaining embedding geometry to the model's inductive bias~\citep{steck2024cosine}.
Sparse feedback supplies few such constraints per embedding.
Graph encoders use interaction-graph propagation~\citep{he2020lightgcn}, while contrastive objectives impose constraints on the embedding distribution~\citep{wu2021self,yu2022simgcl}.

For Matrix Factorization (MF)~\citep{koren2009matrix} trained with Bayesian Personalized Ranking (BPR)~\citep{rendle2012bpr}, we study an inductive bias that constrains each user through other users.
To that end, \textbf{P}athwise \textbf{A}lignment for \textbf{C}ollaborative \textbf{F}iltering (PACF) augments the ranking loss with an auxiliary regression shared by all training pairs.
As shown in Fig.~\ref{fig:sc}, PACF places a directional target at sampled positions on the segment between the user and item embeddings of each training pair, pointing toward the item for an observed pair and back toward the user for a sampled negative pair.
This \textbf{pathwise alignment} learns the shared predictor jointly with the embeddings.

PACF sets this target to the user-to-item displacement at every position, with a sign set by the label, and we call this pair of signed targets \textbf{Bilateral Conditional Flow} (BCF).
PACF fits one velocity network to these targets under a flow matching loss~\citep{lipman2022flow} and ties the integrated endpoint to the observed item with a trajectory consistency term.
Both losses update the embeddings, so PACF discards the network after training and ranks candidates with static inner products.
When the network's sensitivity to its state input is bounded, PACF accurately fits users who share an item only if their embeddings are close, so each user draws constraints from the other users of its items.

Our contributions cover the training objective, its analysis, and its empirical evaluation:

\begin{itemize}
\item PACF reframes collaborative filtering as pathwise alignment, in which a Bilateral Conditional Flow regularizer turns scalar feedback into signed directional targets along user-item segments. One shared network fits these targets, and through gradient amortization its training signal stays in the static embeddings used for inner-product ranking.
\item A shared network with bounded state sensitivity provably fits users of the same item accurately only if their embeddings are close, which constrains each user through the others.
\item Against eleven discriminative and generative baselines on seven datasets from four domains, PACF achieves the highest performance on every dataset, with its largest relative gains over MF where feedback is sparsest, while maintaining the efficiency of MF at inference.
\end{itemize}

\section{Related Work}

\paragraph{Collaborative Filtering.}
CF learns user and item embeddings from observed interactions~\citep{shi2014collaborative,koren2021advances}. MF scores pairs by an inner product, and BPR trains these embeddings to rank an observed item above a sampled negative item~\citep{koren2009matrix,bokde2015matrix,rendle2012bpr}. Later methods enrich the scorer, as in Neural Collaborative Filtering (NCF)~\citep{he2017neural}, or the representations, through graph propagation~\citep{he2020lightgcn} and contrastive objectives on augmented views or neighbors~\citep{wu2021self,yu2022simgcl,zhang2025lightccf}. DirectAU aligns the embeddings of observed pairs and spreads them uniformly~\citep{wang2022towards,wang2020understanding}, while SCCF and RecDCL shape the embeddings through contrastive objectives~\citep{wu2024unifying,zhang2024recdcl}. 
Graph-regularized MF and UltraGCN tie the embeddings of related users and items through explicit graph constraints~\citep{rao2015collaborative,mao2021ultragcn}. 
PACF keeps the MF scorer and the BPR loss and adds a displacement regression shared by all training pairs, so one learned fit also constrains each user through the other users who share any of its items.

\paragraph{Flow Matching for Recommendation.}
Flow matching regresses a velocity field onto target velocities along conditional paths between coupled source and target points~\citep{lipman2022flow}, and an ordinary differential equation (ODE) integrates the field into trajectories~\citep{chen2018neural,finlay2020train}.
Each PACF training pair fixes this coupling, so the conditional path becomes the relevant design choice~\citep{peyre2019computational,pooladian2023multisample}.
Generative flow models use the field at inference to transform distributions~\citep{liu2023flow,gat2024discrete}, and recommender systems apply this design to interaction distributions, sequential preferences, and denoised social signals~\citep{li2025preference,huang2025flow}.
FlowCF applies it to binary interaction vectors with a behavior-guided prior~\citep{liu2025flow}.
PACF instead fits the field to signed user-to-item displacements to regularize static embeddings during training, then discards it and scores items by inner products.

\section{Preliminaries}

\subsection{Notations}
Let $\mathcal{U}$ and $\mathcal{I}$ denote the sets of users and items, and let $\mathcal{O}\subseteq\mathcal{U}\times\mathcal{I}$ collect the observed interactions, with $\mathcal{I}_u=\{i:(u,i)\in\mathcal{O}\}$ for user $u$.
We map each entity to a shared $d$-dimensional Euclidean space $\mathcal{M} = \mathbb{R}^d$.
Users and items are parameterized by embeddings $\mathbf{P}\in\mathbb{R}^{|\mathcal{U}|\times d}$ and $\mathbf{Q}\in\mathbb{R}^{|\mathcal{I}|\times d}$ with rows $\mathbf{p}_u$ and $\mathbf{q}_i$.
PACF is based on MF and randomly initializes $\mathbf{P}$ and $\mathbf{Q}$ as ID embedding tables.
Each training pair carries a label $y_{ui}\in\{+1,-1\}$.
A positive item $i^{+}$ of user $u$ satisfies $(u,i^{+})\in\mathcal{O}$ and $y_{ui}=+1$, and a sampled negative item $i^{-}$ is drawn from the unobserved items $\mathcal{I}\setminus\mathcal{I}_u$ with $y_{ui}=-1$. Each positive pair is trained with $K_{-}$ sampled negatives.
The interpolation time $t\in[0,1]$ indexes states on the segment from a user embedding at $t=0$ to an item embedding at $t=1$.

\subsection{Problem Formulation}
CF learns user and item embeddings from sparse feedback in the shared space $\mathcal{M}$.
Their inner product defines the score $s(u,i)$ that ranks all candidate items for each user at inference:
\begin{equation}\label{eq:pred}
    s(u,i) = \langle \mathbf{p}_u, \mathbf{q}_i\rangle.
\end{equation}
BPR~\citep{rendle2012bpr} ranks an observed item $i^+$ above a sampled negative $i^-$ through their score difference and minimizes the expected logistic loss over sampled training triplets $(u,i^+,i^-)$:

\begin{equation}
\label{eq:ranking_loss}
    \mathcal{L}_{\mathrm{rec}}
    =
    \mathbb{E}_{(u,i^+,i^-)}
    \left[ -\log\sigma\big(s(u,i^+)-s(u,i^-)\big) \right],
\end{equation}
where $\sigma$ is the sigmoid.
Eq.~\eqref{eq:ranking_loss} updates each user embedding only through the scores of that user's observed and sampled pairs, so sparse feedback gives the embedding few constraints.

\section{Methodology}
\label{sec:method}

\begin{wrapfigure}[20]{r}{0.6\columnwidth}
  \centering
  \includegraphics[width=1\linewidth]{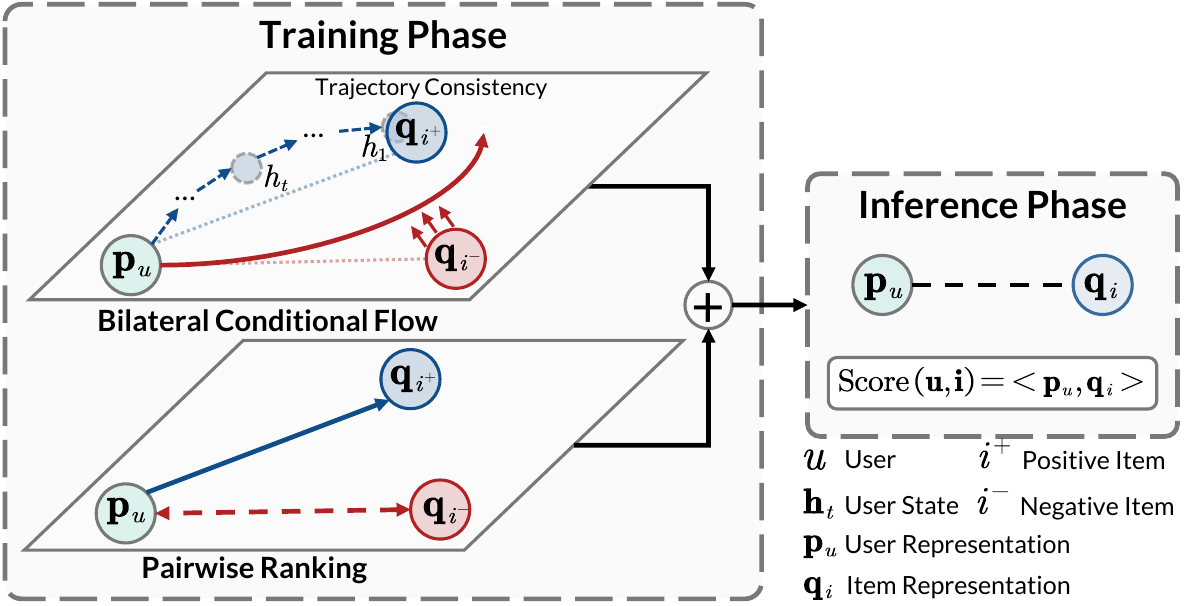}
  \caption{
    Framework of PACF.
    Left: \textbf{Training} fits the shared velocity network to the signed BCF targets at states $\mathbf{h}_t$ sampled on the reference segment.
    A trajectory consistency term on positive pairs pulls the RK4 endpoint $\Phi_1(\mathbf{p}_u)$ toward $\mathbf{q}_{i^+}$, so negative pairs enter the pathwise losses only through the flow matching term.
    Right: \textbf{Inference} discards the network and scores by inner product. PACF thus adds cost only during training and ranks with the same efficiency as MF.
  }
  \label{fig:framework}
\end{wrapfigure}

The recommendation loss $\mathcal{L}_{\mathrm{rec}}$ in Eq.~\eqref{eq:ranking_loss} shapes the user and item embeddings through the scores of observed and sampled pairs.
Fig.~\ref{fig:framework} shows the two losses PACF adds to each user-item segment, namely flow matching toward BCF targets at sampled states and trajectory consistency at the integrated endpoint. Their gradients update the static embeddings that Eq.~\eqref{eq:pred} scores during inference.

\subsection{Bilateral Conditional Flow}
\label{sec:bcf}

For a user $u$ and an item $i$ with label $y_{ui} \in \{+1, -1\}$, the intermediate state on the segment at time $t \in [0, 1]$ is the straight-line interpolation:
\begin{equation}
\label{eq:path}
\mathbf{h}_t \,\equiv\, \mathbf{h}^{(u,i)}_t = (1-t)\,\mathbf{p}_u + t\,\mathbf{q}_i.
\end{equation}
Its velocity $\dot{\mathbf{h}}_t=\mathbf{q}_i-\mathbf{p}_u$ is constant along the path.
The observed label sets the direction and scale of the signed BCF regression target, which is therefore also constant along the segment:
\begin{equation}
\label{eq:target_vector}
\mathbf{u}_t(\mathbf{h}_t \mid y_{ui}) = \mu(y_{ui}) \,(\mathbf{q}_i - \mathbf{p}_u),
\quad \text{where }
\mu(y) = \begin{cases} +1, & y = +1 \\ -\mu_-, & y = -1 \end{cases}.
\end{equation}
In Eq.~\eqref{eq:target_vector}, $\mu_- \in [0, 1]$ scales the target on negative pairs.
The BCF target points from the user embedding toward the item embedding for positive pairs and back toward the user for negative pairs.

\subsection{Pathwise Flow Matching}
\label{sec:flow_reg}

In PACF, one velocity network $\mathbf{v}_{\boldsymbol\theta}$ is shared by all training pairs and regressed onto the BCF target $\mathbf{u}_t(\mathbf{h}_t \mid y_{ui})$ in Eq.~\eqref{eq:target_vector}. On a positive pair this target is the reference-path velocity used by conditional flow matching~\citep{lipman2022flow}, and on a sampled negative pair it reverses the velocity of that pair's own segment and scales it by $\mu_-$. The flow matching loss is the expected squared error:
\begin{equation}
\label{eq:flow_loss}
\mathcal{L}_{\mathrm{flow}} = \mathbb{E}_{t,\,(u,i,y_{ui})}\!\left[
\bigl\| \mathbf{v}_{\boldsymbol\theta}(\mathbf{h}_t, t, \mathbf{q}_i, y_{ui}) - \mathbf{u}_t(\mathbf{h}_t \mid y_{ui}) \bigr\|_2^2
\right].
\end{equation}
In Eq.~\eqref{eq:flow_loss}, $t\sim\Unif[0,1]$ is drawn independently of the training pair, which ranges over each positive pair and its $K_{-}$ sampled negatives.
The target keeps its $(\mathbf{p}_u,\mathbf{q}_i)$ dependence during backpropagation, so $\mathcal{L}_{\mathrm{flow}}$ updates the embeddings as well as $\mathbf{v}_{\boldsymbol\theta}$, which serves only as a training-time regressor.

The velocity network is a three-layer residual multilayer perceptron (MLP) with hidden width $d_h$:
\begin{equation}
\label{eq:resmlp}
\mathbf{z} 
= 
\GELU\bigl(
  \mathbf{W}_1\LN([\mathbf{h}_t;\,\mathbf{q}_i;\,t;\,y_{ui}])
\bigr), 
\quad
\mathbf{v}_{\boldsymbol\theta} 
= 
\mathbf{W}_3\bigl(
  \mathbf{z} + \GELU(\mathbf{W}_2 \,\mathbf{z})
\bigr).
\end{equation}
Layer normalization (LN) acts on the input, and the Gaussian error linear unit (GELU) follows each hidden projection. 
The trainable weights are $\boldsymbol\theta = \{\mathbf{W}_1, \mathbf{W}_2, \mathbf{W}_3\}$.
The skip connection passes the hidden representation $\mathbf{z}$ directly to $\mathbf{W}_3$, and all training pairs share these weight matrices.

\subsection{Trajectory Consistency}
\label{sec:amortize}

The flow matching loss $\mathcal{L}_{\mathrm{flow}}$ in Eq.~\eqref{eq:flow_loss} supervises the field $\mathbf{v}_{\boldsymbol\theta}$ only at sampled reference states, while trajectory consistency integrates this field from $\mathbf{p}_u$ to $t=1$ and supervises the resulting endpoint.
For a positive pair $(u,i^+)$, the trajectory starts at $\widetilde{\mathbf{h}}_0=\mathbf{p}_u$ and follows the neural ODE:

\begin{equation}
\label{eq:node}
\dot{\widetilde{\mathbf{h}}}_t = \mathbf{v}_{\boldsymbol\theta}(\widetilde{\mathbf{h}}_t, t, \mathbf{q}_{i^+}, +1),
\end{equation}
with exact endpoint $\widetilde{\mathbf{h}}_1$. With $(\mathbf{q}_{i^+},+1)$ fixed, the fixed-step fourth-order Runge--Kutta (RK4) solver~\citep{butcher2016numerical} with $T$ steps approximates $\widetilde{\mathbf{h}}_1$ by $\Phi_1(\mathbf{p}_u)$ in all reported experiments.
Each RK4 step evaluates the field once at the current state, twice at midpoint estimates, and once at an endpoint estimate. 
The trajectory consistency loss $\mathcal{L}_{\mathrm{traj}}$ measures the squared distance between this endpoint and the observed item over positive pairs, whose BCF target leads to that item:
\begin{equation}
\label{eq:traj_constraint}
\mathcal{L}_{\mathrm{traj}} = \mathbb{E}_{(u, i^+)}\!\left[\,\big\|\Phi_1(\mathbf{p}_u) - \mathbf{q}_{i^+}\big\|_2^2\,\right].
\end{equation}
The solver backpropagates this endpoint residual through all $T$ RK4 steps to $\mathbf{p}_u$, $\mathbf{q}_{i^+}$, and $\boldsymbol\theta$.

\subsection{Training Objective}
\label{sec:obj}

Adding both pathwise losses to the ranking loss in Eq.~\eqref{eq:ranking_loss} gives the full training objective:
\begin{equation}
\label{eq:obj}
\min_{\{\mathbf{p}_u\},\,\{\mathbf{q}_i\},\,\boldsymbol\theta}
\mathcal{L}_{\mathrm{rec}}
+ \lambda\mathcal{L}_{\mathrm{flow}}
+ \beta\mathcal{L}_{\mathrm{traj}},
\end{equation}
where $\lambda,\beta\geq0$ weight the two auxiliary losses. Setting $\lambda=\beta=0$ recovers BPR and its gradients.

\subsection{Static Scoring through Gradient Amortization}
\label{sec:gradient_amortization}

\textbf{Gradient amortization} writes the fitting information of both auxiliary losses into the static embeddings through their training gradients, and these embeddings keep it after the velocity network and solver are discarded. The network and the RK4 solver thus add cost only during training.

\begin{proposition}[User-Side Gradient Decomposition]
\label{prop:user_gradient_decomposition}
For one training triplet $(u,i^+,i^-)$, let $\ell_{\mathrm{flow}}$ sum its sampled flow losses that depend on $\mathbf{p}_u$. Define $\Delta s=s(u,i^+)-s(u,i^-)$, $\mathbf{r}_u=\Phi_1(\mathbf{p}_u)-\mathbf{q}_{i^+}$, and $\mathbf{J}_{\Phi}(\mathbf{p}_u)=\partial\Phi_1/\partial\mathbf{p}_u$.
With the item embeddings and $\boldsymbol\theta$ fixed, the per-sample objective $\ell=-\log\sigma(\Delta s)+\lambda\ell_{\mathrm{flow}}+\beta\|\mathbf{r}_u\|_2^2$ has a negative gradient with one additive term per loss:
\begin{equation}
\label{eq:grad_traj_refined}
-\nabla_{\mathbf{p}_u}\ell
=
\sigma(-\Delta s)(\mathbf{q}_{i^+}-\mathbf{q}_{i^-})
-\lambda\nabla_{\mathbf{p}_u}\ell_{\mathrm{flow}}
-2\beta\mathbf{J}_{\Phi}(\mathbf{p}_u)^\top\mathbf{r}_u.
\end{equation}
\end{proposition}

The first term moves $\mathbf{p}_u$ along $\mathbf{q}_{i^+}-\mathbf{q}_{i^-}$ with $\sigma(-\Delta s)$, which grows as the score difference falls. 
The flow term differentiates the sampled regression errors through the state and the signed target, and the transpose of $\mathbf{J}_{\Phi}(\mathbf{p}_u)$, which describes how the initial user embedding moves the RK4 endpoint, returns $\mathbf{r}_u$ to $\mathbf{p}_u$. 
Optimization forms each embedding step from their sum, which reduces to the BPR gradient when the flow and endpoint residuals both vanish, so residuals carry all pathwise signal.

\section{Algorithmic Analysis}
\label{sec:analysis}

Flow matching updates the embeddings in proportion to the regression residual, through both the state that the network reads and the signed target. Written as a function of the state, the BCF target is $\mu(y)(\mathbf{q}_i-\mathbf{h}_t)/(1-t)$ for $t<1$, so its sensitivity to the state grows as $|\mu(y)|/(1-t)$ toward the item. One network with bounded sensitivity serves every user of the item, so users whose embeddings differ keep residuals where their states converge, and these residuals feed back into their embeddings.

\subsection{Shared Fitting at the Item End}
\label{sec:item_end}

For two users $u$ and $v$ paired with item $i$ under label $y$, the reference states at time $t$ differ by $(1-t)(\mathbf{p}_u-\mathbf{p}_v)$ and converge to $\mathbf{q}_i$, while the signed targets keep the difference $\mu(y)(\mathbf{p}_v-\mathbf{p}_u)$. A field with state sensitivity $L_{\mathbf h}$, its Lipschitz constant in $\mathbf{h}$, changes its output between the two states by at most $L_{\mathbf h}(1-t)\|\mathbf{p}_u-\mathbf{p}_v\|_2$, so the remaining target difference appears in their residuals.

\begin{proposition}[Item-End Compromise]
\label{prop:item_end}
Fix an item $i$, a label $y$, and $t\in[0,1)$.
For a user $u$ paired with $i$ under label $y$, let $\mathbf{e}_u(t)=\mathbf{v}_{\boldsymbol\theta}(\mathbf{h}^{(u,i)}_t,t,\mathbf{q}_i,y)-\mu(y)(\mathbf{q}_i-\mathbf{p}_u)$ be its flow residual.
If $\mathbf{v}_{\boldsymbol\theta}(\cdot,t,\mathbf{q}_i,y)$ is $L_{\mathbf h}$-Lipschitz between the states of two such users $u$ and $v$, their residuals satisfy:
\begin{equation}
\label{eq:item_end_pointwise}
\|\mathbf{e}_u(t)\|_2+\|\mathbf{e}_v(t)\|_2
\geq
\bigl(|\mu(y)|-L_{\mathbf h}(1-t)\bigr)_+\,\|\mathbf{p}_u-\mathbf{p}_v\|_2.
\end{equation}
Here $(x)_+=\max\{x,0\}$. 
If the Lipschitz condition holds pairwise on a set $\mathcal{S}$ of such users for all $t\in[0,1)$, with a common constant $L_{\mathbf h}>0$ satisfying $L_{\mathbf h}\geq|\mu(y)|$, then $t\sim\Unif[0,1]$ gives the following bound, where $\bar{\mathbf{p}}_{\mathcal{S}}$ is the mean of the users in $\mathcal{S}$, so the right side measures their dispersion:
\begin{equation}
\label{eq:item_end_dispersion}
\mathbb{E}_t\sum_{u\in\mathcal{S}}\|\mathbf{e}_u(t)\|_2^2
\geq
\frac{|\mu(y)|^3}{6L_{\mathbf h}}\sum_{u\in\mathcal{S}}\|\mathbf{p}_u-\bar{\mathbf{p}}_{\mathcal{S}}\|_2^2.
\end{equation}
\end{proposition}

Prop.~\ref{prop:item_end} gives a necessary condition for small flow error at fixed state sensitivity, and its bound falls as these embeddings move together or as $L_{\mathbf h}$ grows.
For $t<1$ the complete input determines $\mathbf{p}_u$, so a field of unrestricted sensitivity could fit each user separately, and the coupling comes from the bounded state sensitivity of the shared network.
The state sensitivity $|\mu(y)|/(1-t)$ of the target exceeds $L_{\mathbf h}$ exactly where the coefficient in Eq.~\eqref{eq:item_end_pointwise} is positive. For $y=+1$, $\mathcal{S}$ holds observed users of item $i$, and for $y=-1$, users for whom $i$ is a sampled negative. 
App.~\ref{app:proofs} proves both results.

\subsection{Complexity}
\label{sec:complexity}

\textbf{Space.}
The embedding tables $\mathbf{P}$ and $\mathbf{Q}$ contain $O((|\mathcal{U}|+|\mathcal{I}|)d)$ parameters. The input and output projections of the residual MLP in Eq.~\eqref{eq:resmlp} contribute $O(dd_h)$ parameters, and its hidden projection contributes $O(d_h^2)$. Total parameter storage is $O((|\mathcal{U}|+|\mathcal{I}|)d+d_h(d+d_h))$. The shared field adds a cost independent of user and item counts. Training also stores optimizer states and activations for differentiating the field and RK4 steps. Deployment retains only the embedding tables.

\textbf{Training time.}
One field evaluation costs $O(d_h(d+d_h))$. For a batch of $B$ positive pairs with $K_{-}$ negatives per positive, flow matching evaluates $B(1+K_{-})$ sampled states. Trajectory consistency integrates only the positive pairs, using four field evaluations per RK4 step and $T$ steps per pair. These losses therefore require $B(1+K_{-}+4T)$ field evaluations and cost $O((1+K_{-}+4T)B\,d_h(d+d_h))$ beyond the BPR loss. Backpropagation through this fixed computation graph has the same asymptotic order. For example, $K_{-}=4$ and $T=4$ give 21 field evaluations per positive pair. For fixed model dimensions, $K_{-}$, and $T$, the cost of one pass is linear in the number of training positive pairs.

\textbf{Inference time.}
PACF discards the velocity network and solver after training. Scoring uses the inner product of the static embeddings in Eq.~\eqref{eq:pred}, costing $O(d)$ per user-item pair and $O(|\mathcal{I}|d)$ per user. Deployment therefore uses the same scoring operations as MF at the same embedding dimension.

\section{Experiments}
\label{sec:experiments}

\subsection{Baselines}
We compare PACF with eleven baselines from two families. \textit{Discriminative Models}: \textbf{BPR-MF}~\citep{rendle2012bpr} learns pairwise rankings, \textbf{NCF} and \textbf{NeuMF}~\citep{he2017neural} use a multi-layer perceptron and a hybrid of linear and non-linear branches as the interaction function, \textbf{LightGCN}~\citep{he2020lightgcn} propagates signals through neighborhood aggregation, \textbf{SGL}~\citep{wu2021self} and \textbf{SimGCL}~\citep{yu2022simgcl} add contrastive objectives through graph augmentation and embedding noise, and \textbf{LightCCF}~\citep{zhang2025lightccf} learns neighborhood aggregation through a contrastive objective. \textit{Generative Models}: \textbf{MultiVAE}~\citep{liang2018variational} and \textbf{RecVAE}~\citep{shenbin2020recvae} model implicit feedback with variational autoencoders, \textbf{DiffRec}~\citep{wang2023diffusion} recovers interaction vectors from noise by diffusion, and \textbf{FlowCF}~\citep{liu2025flow} fits binary interaction vectors by flow matching with a behavior-guided prior, which is the closest flow-based baseline to PACF.

\subsection{Datasets}
\begin{wraptable}[15]{r}{0.55\columnwidth}
\centering
\caption{
  Statistics of the seven datasets used in our experiments.
  Sparsity is the fraction of user-item pairs without any observed interaction in the data.}
\label{tab:datasets}
\resizebox{1\linewidth}{!}{
\begin{tabular}{lcccc}
\toprule
Dataset  & \#Users & \#Items & \#Interactions & Sparsity \\ \midrule
ML       & 610     & 3,650    & 90,274     & 95.95\%  \\
Food     & 5,990    & 1,125    & 36,482     & 99.46\%  \\
Dance    & 10,231   & 1,676    & 80,086     & 99.53\%  \\
Yelp     & 31,668   & 38,048   & 1,561,406   & 99.87\% \\
Baby     & 19,445   & 7,050    & 161,276    & 99.88\%  \\
Sports   & 35,598   & 18,357   & 296,337    & 99.95\%  \\
Clothing & 39,387   & 23,033   & 278,677    & 99.97\%  \\ \bottomrule
\end{tabular}
}
\end{wraptable}
We evaluate PACF on seven public datasets from four domains, and Tab.~\ref{tab:datasets} summarizes their statistics.
For movies, we use \textbf{MovieLens-latest-small} (ML)\footnote{https://grouplens.org/datasets/movielens/latest/}~\citep{harper2015movielens}.
Food and Dance denote Bili\_Food and Bili\_Dance from \textbf{NineRec}\footnote{https://github.com/westlake-repl/NineRec}~\citep{zhang2024ninerec}, used for short-video recommendation.
The \textbf{Yelp2018} (Yelp)\footnote{https://www.kaggle.com/datasets/yelp-dataset/yelp-dataset} benchmark covers business recommendation~\citep{he2020lightgcn}.
From the 2014 \textbf{Amazon review data}\footnote{https://cseweb.ucsd.edu/\textasciitilde jmcauley/datasets/amazon/links.html}~\citep{mcauley2015image,he2016ups}, we use the 5-core subsets of the Baby, ``Sports and Outdoors'', and ``Clothing, Shoes and Jewelry'' categories for e-commerce, denoted Baby, Sports, and Clothing.
We treat every observed interaction as positive implicit feedback, and pair each positive item with $K_{-}=4$ negative items sampled uniformly from the unobserved pool.
All datasets keep users and items with at least five observed interactions.

\subsection{Experimental Setup}
\label{sec:setup}

\textbf{Evaluation Protocols and Metrics.}
For leave-one-out evaluation~\citep{he2017neural}, we reserve each user's latest interaction for testing and the second latest for validation.
Over all candidate items, Recall@10 measures whether the held-out item enters the top ten, and Normalized Discounted Cumulative Gain (NDCG@10) also rewards higher placement.
Main-comparison scores are five-seed means. Validation NDCG@10 selects each ablation control in its own search, so every control is tuned for its design choice, and the PACF entries in these ablations repeat the main comparison.

\textbf{Hyperparameter Configuration.}
For the main comparison, we use Bayesian optimization with a budget of 500 trials for each method and dataset.
We tune the learning rate over $[10^{-5}, 10^{-1}]$ on a logarithmic scale, and weight decay over $[10^{-8}, 10^{-4}]$.
The embedding dimension $d$ and the hidden width $d_h$ of the velocity network are both chosen from $\{64, 128, 256, 512\}$.
For the flow dynamics, we search for the number of integration steps $T \in \{1, 4, 8, 16\}$.
We tune the regularization weights $\lambda \in [0.1, 1.0]$ and $\beta \in [0, 0.5]$, and the negative target scale $\mu_- \in [0, 1]$.
We tune each baseline over the ranges recommended in its original paper.
App.~\ref{sec:more_exp} gives the ablation protocols, initialization, hardware, selected settings, training and inference costs, and five-seed standard deviations.

\begin{table*}[!t]
\centering
\caption{
Comparison on seven datasets, measured by Recall@10 (Recall) and NDCG@10 (NDCG) in percent. 
\textbf{Bold} and \underline{underline} mark the best and second-best results for each dataset and metric. PACF has the best result in every column, and its relative gain over BPR-MF grows as sparsity rises.
}
\label{tab:main}
\begingroup
\fontsize{8.5}{9.5}\selectfont
\setlength{\tabcolsep}{0pt}
\newcommand{\mainmetrichead}[1]{\makebox[\dimexpr.90\linewidth/14\relax][c]{#1}}
\begin{tabular}{@{}l*{14}{c}@{}}
\toprule
         & \multicolumn{2}{c}{\textbf{ML}}   & \multicolumn{2}{c}{\textbf{Food}} & \multicolumn{2}{c}{\textbf{Dance}} & \multicolumn{2}{c}{\textbf{Yelp}} & \multicolumn{2}{c}{\textbf{Baby}} & \multicolumn{2}{c}{\textbf{Sports}} & \multicolumn{2}{c}{\textbf{Clothing}} \\
\makebox[.10\linewidth][l]{Method} & \mainmetrichead{Recall}          & \mainmetrichead{NDCG}            & \mainmetrichead{Recall}          & \mainmetrichead{NDCG}            & \mainmetrichead{Recall}           & \mainmetrichead{NDCG}            & \mainmetrichead{Recall}          & \mainmetrichead{NDCG}            & \mainmetrichead{Recall}          & \mainmetrichead{NDCG}            & \mainmetrichead{Recall}           & \mainmetrichead{NDCG}             & \mainmetrichead{Recall}            & \mainmetrichead{NDCG}              \\ \midrule
BPR-MF   & 5.59 & 2.66 & 6.91 & 3.36 & 6.62 & 3.29 & 2.77 & 1.37 & 1.42 & 0.67 & 1.48 & 0.74 & 0.64 & 0.30 \\
NeuMF    & 6.18 & 3.11 & 7.38 & 3.55 & 6.66 & 3.31 & 3.11 & 1.54 & 1.43 & 0.68 & 1.49 & 0.74 & 0.65 & 0.31 \\
NCF      & 6.39 & 3.24 & 7.57 & 3.63 & 6.72 & 3.37 & 3.19 & 1.60 & 1.68 & 0.82 & 1.48 & 0.75 & 0.68 & 0.32 \\
LightGCN & 6.52 & 3.16 & 7.44 & 3.60 & 6.19 & 3.09 & 3.26 & 1.62 & 1.37 & 0.66 & 1.34 & 0.67 & 0.58 & 0.27 \\
SGL      & 6.72 & 3.35 & 9.08 & 4.48 & 9.23 & 4.83 & 4.23 & 2.10 & 2.16 & 1.11 & 2.29 & 1.19 & 1.30 & 0.69 \\
SimGCL   & 6.13 & 2.98 & 9.46 & \underline{4.71} & 10.17 & 5.17 & 4.40 & 2.19 & \underline{2.58} & 1.32 & 3.45 & 1.82 & 2.56 & 1.35 \\
LightCCF & \underline{7.38} & \underline{3.56} & \underline{9.56} & 4.68 & \underline{10.21} & \underline{5.18} & \underline{4.46} & \underline{2.21} & 2.50 & \underline{1.40} & \underline{3.48} & \underline{1.91} & \underline{2.67} & \underline{1.47} \\ \midrule
MultiVAE & 6.27 & 3.09 & 8.36 & 4.09 & 8.37 & 4.24 & 3.54 & 1.77 & 1.67 & 0.80 & 1.36 & 0.69 & 1.49 & 0.77 \\
RecVAE   & 6.82 & 3.14 & 8.02 & 3.93 & 8.74 & 4.56 & 3.62 & 1.81 & 1.61 & 0.80 & 1.60 & 0.81 & 1.47 & 0.79 \\
DiffRec  & 6.82 & 3.54 & 7.34 & 3.53 & 8.65 & 4.26 & 3.08 & 1.53 & 1.57 & 0.72 & 1.58 & 0.85 & 0.65 & 0.30 \\
FlowCF   & 7.31 & 3.45 & 8.36 & 4.31 & 9.05 & 4.79 & 3.56 & 1.72 & 1.79 & 0.91 & 1.98 & 1.08 & 1.22 & 0.66 \\ \midrule
PACF     & \textbf{7.48} & \textbf{3.73} & \textbf{9.75} & \textbf{4.79} & \textbf{10.32} & \textbf{5.25} & \textbf{4.49} & \textbf{2.23} & \textbf{2.63} & \textbf{1.42} & \textbf{3.69} & \textbf{1.97} & \textbf{2.88} & \textbf{1.59} \\ \bottomrule
\end{tabular}
\endgroup
\end{table*}

\subsection{Performance Analysis}
\label{sec:performance}

PACF achieves the best result on every dataset in Tab.~\ref{tab:main}. It keeps the scorer and the ranking loss of BPR-MF, so this advantage lies entirely in the embeddings that its pathwise losses train. LightCCF or SimGCL holds the second-best result in every column, and both share information across users through graph aggregation and a contrastive objective. PACF exceeds them with the plain MF scorer, so its cross-user constraints act only during training and the deployed model stays an inner-product scorer. Every generative baseline trails both methods on the three sparse Amazon datasets.

The gain over BPR-MF rises steadily from ML, the densest dataset, to Clothing, the sparsest. BPR constrains a user only through its own interactions, which sparsity removes, while the shared fit draws on every user of the same items. Clothing also gives PACF its widest relative margin over the strongest baseline. PACF adds constraints rather than scorer capacity, and on sparse data the constraints matter more. NeuMF adds a learned nonlinear interaction to MF and stays close to BPR-MF on the three Amazon datasets, where PACF, with the plain inner product, leads by a wide margin. PACF also exceeds FlowCF on every dataset, and unlike FlowCF its flow regularizes the same user and item embeddings that the inner product later scores for every candidate item.

\subsection{Ablation Study}
\label{sec:ablation}

\subsubsection{Comparison with Endpoint Regularization}
\label{sec:ablation_endpoint}

\begin{wraptable}[14]{r}{0.42\textwidth}
  \centering
  \caption{
  Comparison of PACF with BPR and two endpoint penalties on ML and Food, measured by Recall@10 (Recall) and NDCG@10 (NDCG) in percent. PACF regresses one shared field along each user-item segment, while both penalties act only on endpoint distances.
  }
  \label{tab:endpoint_controls}
  \resizebox{1\linewidth}{!}{
  \begin{tabular}{llcccc}
  \toprule
  Dataset               & Metric & BPR  & EP$^{+}$ & EP$^{\pm}$ & PACF \\ \midrule
  \multirow{2}{*}{ML}   & Recall & 6.06 & 6.89     & 6.72       & 7.48 \\
                      & NDCG   & 3.05 & 3.14     & 3.28       & 3.73 \\ \midrule
  \multirow{2}{*}{Food} & Recall & 7.45 & 8.94     & 9.19       & 9.75 \\
                      & NDCG   & 3.60 & 4.45     & 4.59       & 4.79 \\ \bottomrule
  \end{tabular}
  }
  \end{wraptable}

Penalizing the length of each user-to-item displacement is the most direct way to use it, and Tab.~\ref{tab:endpoint_controls} tests this alternative. \textbf{EP$^{+}$} adds the squared distance of each positive pair to \textbf{BPR}, and \textbf{EP$^{\pm}$} also adds a margin against sampled negatives~\citep{hsieh2017collaborative}. Both penalties improve on BPR, so pulling users toward their items already adds useful constraints to MF.

On positive pairs, EP$^{+}$ is the special case of PACF in which the network predicts nothing. The flow error of a positive pair is then exactly its squared distance, and the integrated endpoint stays at the user embedding, so both pathwise losses collapse to the EP$^{+}$ penalty.
The field changes where each user is pulled. EP$^{+}$ pulls every user straight toward each of its items, one pair at a time, while PACF pulls through the residual of one shared fit. Near an item, the states of all its users nearly coincide, so the shared network predicts roughly their average displacement, and the residual of each user becomes its offset from the average of its co-users. This residual pulls the user toward that average and shrinks the spread of users around the item that Prop.~\ref{prop:item_end} bounds from below.

Users who share an item thus move toward one another, so an item adopted by one of them scores higher for the others. EP$^{\pm}$ adds a push away from sampled negatives within each pair, while the negative pairs of PACF join the same shared fit through their reversed and scaled targets.

\begin{figure}[!t]
\centering
\begin{minipage}[t]{0.475\linewidth}
  \centering
  \includegraphics[width=\linewidth]{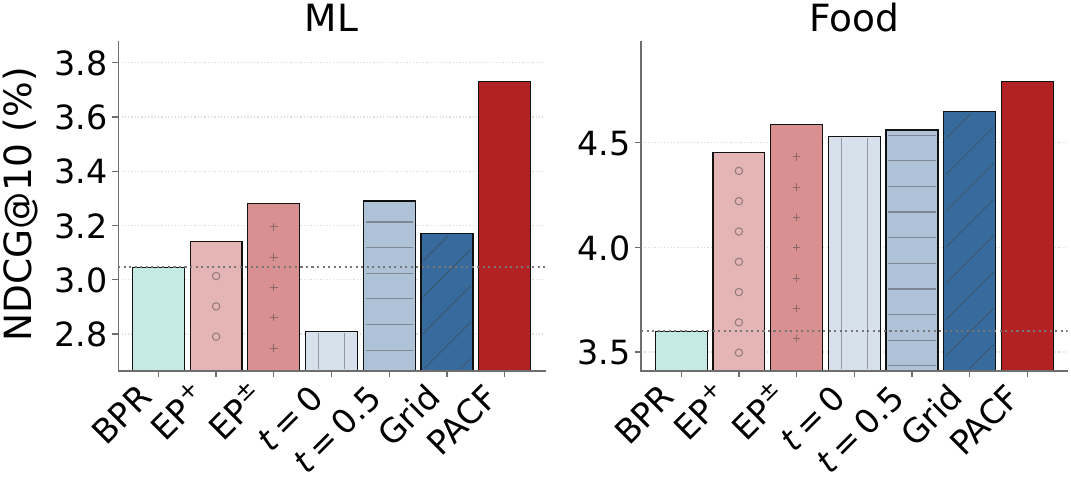}
  \caption{Sampling supports of $t$ in $\mathcal{L}_{\mathrm{flow}}$ on ML and Food, shown with BPR and both endpoint penalties and measured by NDCG@10 in percent. The dotted line marks the BPR level. PACF samples $t$ uniformly over the segment.}
  \label{fig:supervision_density}
\end{minipage}\hfill
\begin{minipage}[t]{0.505\linewidth}
  \centering
  \includegraphics[width=\linewidth]{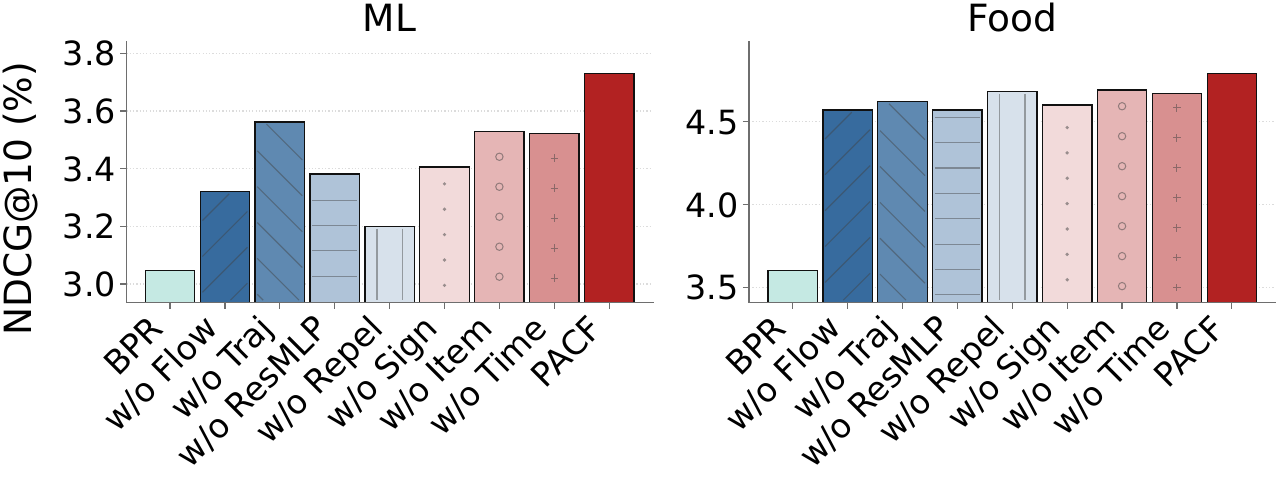}
  \caption{Removal of one PACF component at a time on ML and Food, measured by NDCG@10 in percent against BPR. PACF feeds the complete input $(\mathbf{h}_t,\mathbf{q}_i,t,y_{ui})$ to one residual field and trains it with both pathwise losses on all training pairs.}
  \label{fig:component_ablation}
\end{minipage}
\end{figure}

\subsubsection{Sampling Positions Along the User-Item Path}
\label{sec:ablation_density}

Fig.~\ref{fig:supervision_density} fixes the time of the reference state in Eq.~\eqref{eq:flow_loss}, where \textbf{$t{=}0$} places it at the user end, \textbf{$t{=}0.5$} at the midpoint, and \textbf{Grid} at an anchor drawn from three evenly spaced interior points, while PACF samples the whole segment. Moving the support from the user end to the midpoint, nearer the item, helps.

The BCF target is constant along a segment, so the support only chooses the state the network reads. At the user end every user presents its own embedding, and the network fits each displacement separately. Toward the item, the states of co-users converge while their targets stay apart. At the midpoint these states lie half as far apart as the embeddings. Next to the item, the target changes with the state faster than the network can follow, and Prop.~\ref{prop:item_end} bounds their residuals from below by the spread of these users. Uniform sampling keeps a fixed share of reference states in this region.

\subsubsection{Choice of Flow-Matching Target}
\label{sec:ablation_target}

\begin{wrapfigure}[15]{r}{0.5\columnwidth}
  \centering
  \includegraphics[width=\linewidth]{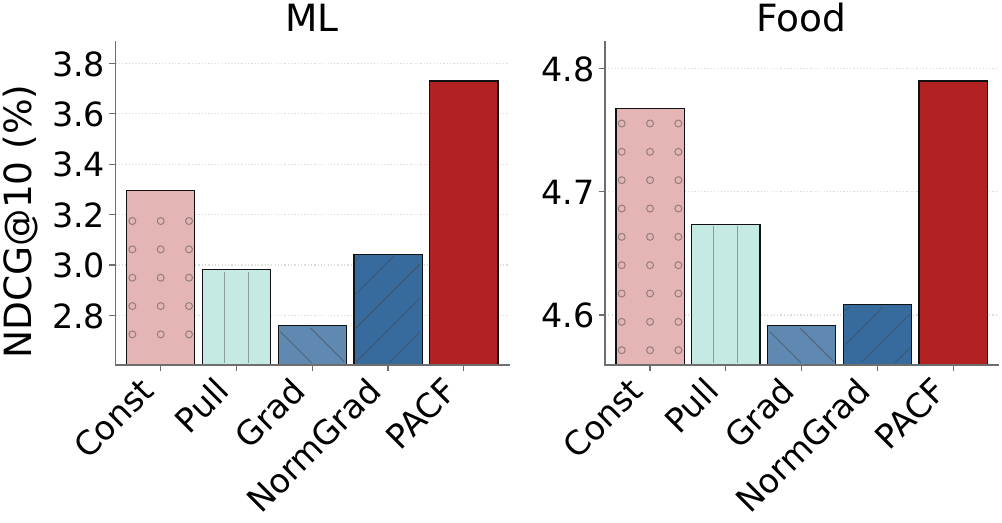}
  \caption{Comparison of flow-matching targets on ML and Food, measured by NDCG@10 in percent. The four target variants share one tuning protocol. PACF reports the separately tuned main result.}
  \label{fig:target_structure_ndcg}
\end{wrapfigure}

Fig.~\ref{fig:target_structure_ndcg} compares four targets for $\mathcal{L}_{\mathrm{flow}}$. \textbf{Const} uses the signed displacement in Eq.~\eqref{eq:target_vector}. \textbf{Pull} points from the sampled state toward the item under both labels. \textbf{Grad} uses the signed score gradient with respect to the user embedding, and \textbf{NormGrad} normalizes it to unit length.

Const gives the highest control NDCG on both datasets. Const retains each user's full displacement as its reference state approaches the item. Users paired with the same item and label present nearby states with distinct targets, creating the residuals that Prop.~\ref{prop:item_end} bounds. Grad and NormGrad assign these users one common target, while Pull makes both labels attractive. The target therefore determines which user differences the shared field must reconcile and how the two labels shape the embeddings.

\subsubsection{Ablation Study of Components}
\label{sec:ablation_components}

Fig.~\ref{fig:component_ablation} removes one component of PACF at a time. The network needs the state, the item, and the time together, because the time converts the remaining distance to the item back into the full displacement. \textbf{w/o Item} removes the item input, so the network must infer the item from a state that mixes user and item, and \textbf{w/o Time} removes the time, so one remaining displacement fits a nearby user early and a distant user late. Either way the residuals encode missing inputs instead of user differences. \textbf{w/o ResMLP}, an equally sized plain MLP, gives the largest drop among the field designs. The BCF target is linear in the embeddings up to the sign $\mu(y)$, so the residual skip carries the linear displacement, and the nonlinear branch absorbs the sign and the correction where co-user segments cross.

The label tells the network which way the target points. \textbf{w/o Sign} zeroes the label, so a positive and a negative pair near the same item present similar inputs with opposite targets. \textbf{w/o Repel} drops negative-pair regression, so sampled negatives act only through BPR.
\textbf{w/o Flow} drops the fit of velocities at single states, and \textbf{w/o Traj} drops the check of where the composed velocities lead, along which local errors accumulate or cancel. Tab.~\ref{tab:trajectory_controls} retunes both removals over five seeds, and each loss alone recovers most of the gain of PACF over BPR-MF while both stay below the full model.

\subsubsection{Inference Readout}
\label{sec:inference_readout}

\begin{wraptable}[13]{R}{0.55\columnwidth}
  \centering
  \caption{Static and endpoint readouts of PACF on ML and Food, measured by Recall@10 and NDCG@10 in percent as five-seed mean $\pm$ std. The static inner product ranks best in every cell, so PACF ranks with the static MF scorer that its two pathwise losses train.}
  \label{tab:inference_readout}
  \small
  \setlength{\tabcolsep}{1.3pt}
  \begin{tabular}{@{}lcccc@{}}
  \toprule
  & \multicolumn{2}{c}{ML} & \multicolumn{2}{c}{Food} \\
  Readout & Recall & NDCG & Recall & NDCG \\ \midrule
  $a$ & $7.48\pm0.41$ & $3.73\pm0.03$ & $9.75\pm0.08$ & $4.79\pm0.05$ \\
  $a_{\Phi}$ & $0.59\pm0.19$ & $0.23\pm0.08$ & $2.08\pm0.06$ & $0.95\pm0.01$ \\
  $a_{\mathrm{dist}}$ & $1.05\pm0.25$ & $0.50\pm0.15$ & $4.68\pm0.58$ & $2.42\pm0.26$ \\ \bottomrule
  \end{tabular}
\end{wraptable}

We compare three readouts of PACF. The static readout $a$ is the inner product of the embeddings and repeats the PACF result of Tab.~\ref{tab:main}. The endpoint readouts integrate the network from the user toward each candidate under a positive label on separately trained checkpoints, one per seed, and $a_{\Phi}$ and $a_{\mathrm{dist}}$ take the inner product and the negative distance between this endpoint and the candidate.

The static readout ranks best in every cell of Tab.~\ref{tab:inference_readout}. Trajectory consistency trains each trajectory to reach its conditioning item, so the better PACF fits it, the closer every endpoint lies to the candidate. $a_{\Phi}$ then approaches the squared norm of the candidate, $a_{\mathrm{dist}}$ approaches one value for all candidates, and only the residual user information ranks candidates. $a_{\Phi}$, which adds the item norm, thus ranks below $a_{\mathrm{dist}}$. Eq.~\eqref{eq:grad_traj_refined} returns this residual to $\mathbf{p}_u$, so the trajectory only trains the static readout.

\subsection{Robustness Analysis}
\label{sec:robustness}

\begin{figure}[!t]
\centering
\begin{subfigure}[b]{0.495\linewidth}
  \includegraphics[width=\linewidth]{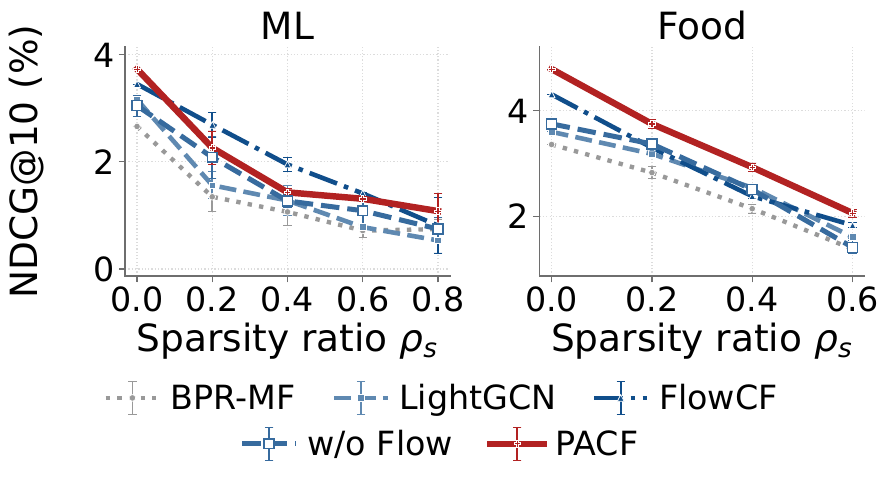}
  \caption{Sparsity Sweep.}
  \label{fig:robustness_sparsity}
\end{subfigure}\hfill
\begin{subfigure}[b]{0.495\linewidth}
  \includegraphics[width=\linewidth]{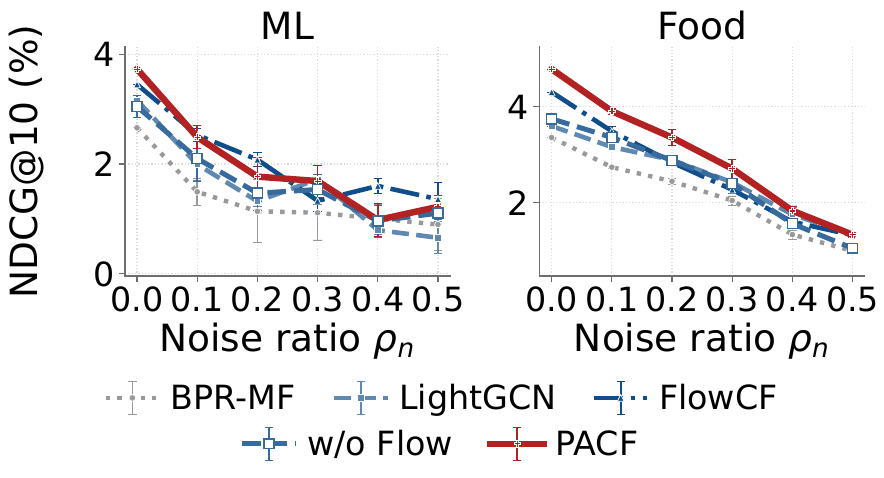}
  \caption{Replacement Noise Sweep.}
  \label{fig:robustness_noise}
\end{subfigure}
\caption{
Robustness of PACF and baselines to interaction removal in panel~a and item replacement in panel~b on ML and Food, measured by NDCG@10 in percent. Error bars show seed standard deviations. PACF stays above \textbf{w/o Flow} at every level and leads all compared models on Food.
  }
\label{fig:robustness_main}
\end{figure}

Fig.~\ref{fig:robustness_main} perturbs only the training interactions, by removing a fraction $\rho_s$ of them or replacing the item of a fraction $\rho_n$ with a random unobserved item, and all models keep their clean-data configurations.

Removal thins the feedback of most users, and PACF stays above \textbf{w/o Flow} at every removal level. PACF leads all compared models on Food and regains the lead on ML at the strongest removal ratio $\rho_s$. Every retained pair still enters the shared regression that couples the users of its item.

Replacement corrupts the segment, BCF target, and trajectory endpoint of each affected pair. PACF stays above \textbf{w/o Flow} at every noise level on both datasets and leads throughout Food. On ML, the densest dataset, PACF gains the least over BPR in the main comparison and loses its lead at most noise levels. With about 150 interactions per user against about six on Food, dense feedback already constrains each ML user through its own observed pairs, leaving less for other users to add.

\section{Conclusion}

For sparse implicit feedback, PACF augments BPR-MF with Bilateral Conditional Flow regularization, in which one shared velocity network regresses signed displacement targets along user-item segments and a trajectory consistency term ties the integrated endpoint of each observed pair to its item. Through gradient amortization, both losses shape the static embeddings. Our analysis shows how residuals drive these updates and gives a necessary condition at fixed state sensitivity that links the flow error of users sharing an item and label to their dispersion. PACF achieves the highest performance on all seven datasets while keeping the MF scorer and gaining most on sparse data.

\subsection*{AI use statement}
Generative AI tools assisted with literature search, manuscript writing, review, and LaTeX formatting.
The authors reviewed and verified all content and take responsibility for the manuscript.

\subsection*{Reproducibility statement}
Sec.~\ref{sec:setup} and Appendix~\ref{sec:implementation_environment} provide the experimental setup and implementation details.
Appendix~\ref{sec:training_algorithm} gives the training algorithm, and Appendix~\ref{app:proofs} presents the theoretical assumptions and proofs.

\bibliography{ref}
\bibliographystyle{pacf_preprint}

\clearpage
\appendix
\suppressfloats[t]

\section{Limitations and Future Work}\label{sec:limitations}

\textbf{Broader Impact.}
PACF serves recommendations with the MF scorer and discards the velocity network after training, and Tab.~\ref{tab:efficiency} reports its measured latency and memory. PACF learns from observed interactions, so a deployment should evaluate the exposure and subgroup-fairness effects that collaborative filtering inherits from these data and mitigate them where they arise.

\textbf{Design Scope.}
PACF uses straight reference paths in a Euclidean embedding space. By App.~\ref{sec:ot_path_proof}, this choice uniquely minimizes Euclidean kinetic energy. It gives the constant signed target $\mu(y_{ui})(\mathbf{q}_i-\mathbf{p}_u)$, and lets the complete flow input identify $\mathbf{p}_u=(\mathbf{h}_t-t\mathbf{q}_i)/(1-t)$ for $t<1$. The gradient and item-end analyses treat free user and item embeddings, matching the deployed MF scorer.

\textbf{Future Work.}
Non-Euclidean path families and state-dependent targets widen this design space, and each changes the target function in Eq.~\eqref{eq:target_function} on which the analysis rests. A graph encoder ties each embedding to its neighbors, so extending the analysis to graph-based recommenders requires joint gradients through the propagation layers, the flow map, and the shared field. Lower-order solvers can reduce training cost, and comparing their trajectory error and ranking accuracy with RK4 would quantify the role of the final stage at $t=1$ in Remark~\ref{rem:traj_endpoint}. Prop.~\ref{prop:item_end} gives a necessary condition that shared fitting at the item end places on the embeddings, and characterizing the full embedding regularizer induced by this coupling is a natural next step toward a sufficient condition.

\section{Method Details}

\subsection{Training Algorithm}\label{sec:training_algorithm}

Alg.~\ref{alg:training} draws a mini-batch $\mathcal{B}=\{(u,i^+,\{i^-_k\}_{k=1}^{K_{-}})\}$ and computes the three losses in Eq.~\eqref{eq:obj}. Adam jointly updates the velocity network and static embeddings. Inference uses only the embeddings.

\begin{algorithm}[t]
   \caption{Algorithm of PACF.}
   \label{alg:training}
\begin{algorithmic}[1]
   \State {\bfseries Input:} Dataset $\mathcal{O}$, batch $B$, weights $\lambda,\beta$, negative scale $\mu_-$, RK4 steps $T$, negatives $K_{-}$.
   \State {\bfseries Initialize:} Embeddings $\mathbf{P}, \mathbf{Q}$ and velocity network $\mathbf{v}_{\boldsymbol\theta}$.
   \State {\bfseries Output:} Trained embeddings $\mathbf{P}, \mathbf{Q}$.
   \While{not converged}
       \State Sample $B$ positive pairs $(u,i^+)$ from $\mathcal{O}$.
       \State Draw $K_{-}$ uniform negatives $i^-_k\in\mathcal{I}\setminus\mathcal{I}_u$ for each pair.

       \State \textcolor{gray}{// 1. Pairwise ranking loss}
       \State Compute $\mathcal{L}_{\mathrm{rec}}$ via Eq.~\eqref{eq:ranking_loss}.

       \State \textcolor{gray}{// 2. Bilateral flow matching}
       \State Sample times $t \sim \Unif[0,1]$ of size $B$.
       \For{$(u, i, y) \in \{(u, i^+, +1)\} \cup \{(u, i^-_k, -1)\}_{k=1}^{K_{-}}$}
           \State Form $\mathbf{h}_t$ via Eq.~\eqref{eq:path} and $\mathbf{u}_t$ via Eq.~\eqref{eq:target_vector}.
           \State Accumulate $\mathcal{L}_{\mathrm{flow}}$ via Eq.~\eqref{eq:flow_loss}.
       \EndFor

       \State \textcolor{gray}{// 3. Trajectory consistency}
       \State Integrate with RK4: $\Phi_1(\mathbf{p}_u) \leftarrow \RKfour(\mathbf{p}_u, \mathbf{q}_{i^+}, \mathbf{v}_{\boldsymbol\theta}, +1, T)$.
       \State Compute $\mathcal{L}_{\mathrm{traj}}$ via Eq.~\eqref{eq:traj_constraint} for the positive item $i^+$.

       \State \textcolor{gray}{// Joint update}
       \State Update $\{\mathbf{p}_u\}, \{\mathbf{q}_i\}, \boldsymbol\theta$ by minimizing Eq.~\eqref{eq:obj}.
   \EndWhile
\end{algorithmic}
\end{algorithm}

For each positive pair, $K_{-}$ sampled negatives contribute separate terms to $\mathcal{L}_{\mathrm{flow}}$. The trajectory loss integrates only the positive condition because its endpoint is an observed training item.

\subsection{Trajectory Integration}\label{sec:rk4_details}
PACF uses RK4 to compute the endpoint in Eq.~\eqref{eq:traj_constraint}. For a fixed positive condition, write the field as:
\begin{equation}
\mathbf{f}_t(\mathbf{h})=\mathbf{v}_{\boldsymbol\theta}(\mathbf{h},t,\mathbf{q}_{i^+},+1).
\end{equation}
One RK4 step over $t\in[0,1]$ combines four evaluations of the shared velocity network $\mathbf{f}_t$:
\begin{equation}
\label{eq:rk4_one_step}
\begin{aligned}
\mathbf{k}_1 &= \mathbf{f}_0(\mathbf{p}_u),\\
\mathbf{k}_2 &= \mathbf{f}_{1/2}(\mathbf{p}_u + \tfrac{1}{2}\mathbf{k}_1),\\
\mathbf{k}_3 &= \mathbf{f}_{1/2}(\mathbf{p}_u + \tfrac{1}{2}\mathbf{k}_2),\\
\mathbf{k}_4 &= \mathbf{f}_1(\mathbf{p}_u + \mathbf{k}_3),\\
\Phi_1(\mathbf{p}_u) &= \mathbf{p}_u + \tfrac{1}{6}(\mathbf{k}_1+2\mathbf{k}_2+2\mathbf{k}_3+\mathbf{k}_4).
\end{aligned}
\end{equation}
The stages $\mathbf{k}_2$ to $\mathbf{k}_4$ evaluate the field at states generated by earlier predictions, which can leave the reference segment. Flow matching penalizes each reference-state residual separately, while trajectory consistency penalizes only the endpoint, where stage errors add up or cancel. The solver sets how far the two terms differ. A forward Euler step gives $\Phi_1(\mathbf{p}_u)=\mathbf{p}_u+\mathbf{f}_0(\mathbf{p}_u)$, so the positive endpoint loss equals flow matching at $t=0$. RK4 composes four predictions per step, and its last stage evaluates the shared field at the item end $t=1$, where Remark~\ref{rem:traj_endpoint} shows it couples the users of each item.

\section{Additional Experiments}
\label{sec:more_exp}

This appendix first fixes the protocols behind every supplementary result. The analyses then follow the pathwise signal of PACF from the learned field to the embedding updates, the role of each component, and the static scorer, and close with robustness, cost, and the selected settings.

\subsection{Experimental Protocols}\label{sec:implementation_environment}

\textbf{Implementation.}
All methods are implemented in PyTorch $2.1$~\citep{paszke2019pytorch} with CUDA $12.1$ on a workstation with an Intel Xeon Gold $6244$ CPU, $395$~GB RAM, and three NVIDIA Quadro RTX $6000$ GPUs with $24$~GB of memory each. Each experiment runs on one GPU, and Bayesian optimization runs its trials in parallel across the three GPUs. PACF uses MF ID embeddings as its backbone. Xavier-uniform initialization sets the user table, the item table, and every layer of $\mathbf{v}_{\boldsymbol\theta}$, every bias starts at zero, and both tables are optimized directly as free model parameters.

\textbf{Main comparison.}
Each method receives 500 Bayesian optimization trials with the same early stopping. PACF searches the space of Sec.~\ref{sec:setup}, and each baseline the ranges of its original paper. PACF and the pairwise discriminative baselines pair each training positive with $K_{-}=4$ uniform negatives from $\mathcal{I}\setminus\mathcal{I}_u$, and the other baselines keep their original objectives and sampling rules.

\textbf{Endpoint controls.}\label{sec:ablation_protocols}
The endpoint controls in Tab.~\ref{tab:endpoint_controls} keep the MF embeddings and the BPR loss and replace the shared regression of PACF with penalties on distances between users and items. For a training triplet $(u,i^+,i^-)$, they use the mean squared coordinate distances $D_\pm=\|\mathbf{p}_u-\mathbf{q}_{i^\pm}\|_2^2/d$ and add endpoint penalties to the BPR ranking loss while keeping the inner-product scorer:
\begin{equation}
\label{eq:endpoint_control_objectives}
\begin{aligned}
\mathcal{L}_{\mathrm{EP}^{+}}
&=\mathcal{L}_{\mathrm{rec}}+\lambda_{\mathrm{E}}\,\mathbb{E}[D_+],\\
\mathcal{L}_{\mathrm{EP}^{\pm}}
&=\mathcal{L}_{\mathrm{rec}}+\lambda_{\mathrm{E}}\,\mathbb{E}\!\left[D_++\gamma_{\mathrm{E}}\max\{0,m_{\mathrm{E}}+D_+-D_-\}\right].
\end{aligned}
\end{equation}
Here $\lambda_{\mathrm{E}}$ weights endpoint regularization, $\gamma_{\mathrm{E}}$ weights the negative margin term, and $m_{\mathrm{E}}$ sets its distance margin, each searched over $[0,1]$. EP$^{+}$ tunes only the weight $\lambda_{\mathrm{E}}$. EP$^{\pm}$ also penalizes every negative item whose distance to the user exceeds that of the positive item by less than $m_{\mathrm{E}}$.

\textbf{Flow-based controls.}
The sampling and component controls regress onto $\mathbf{q}_{i^+}-\mathbf{p}_u$ for positives and $-(\mathbf{q}_{i^-}-\mathbf{p}_u)$ for negatives. They average the positive and negative errors separately over batch entries and embedding coordinates and add the negative error with a weight $\omega_-\in[0,1]$, a loss weight of these controls distinct from the target scale $\mu_-$ of PACF. The sampling controls keep reference-state regression and the RK4 trajectory term and draw one time per training triplet, shared by its positive and negative flow terms. The point controls fix this time at $0$ or $0.5$. Grid draws it uniformly from $\{0.25,0.5,0.75\}$, which matches the number of reference-state evaluations of uniform sampling. These draws set only the flow-matching states, and the trajectory term keeps RK4 integration.

\label{sec:label_control_protocol}\textbf{w/o Sign} sets the label coordinate of the network input in Eq.~\eqref{eq:resmlp} to zero in flow matching and trajectory integration and keeps the signed targets of Eq.~\eqref{eq:target_vector}, so the field receives the label only through its targets. \textbf{w/o Repel} trains with $\mathcal{L}_{\mathrm{rec}}$, positive flow matching, and positive trajectory consistency, so sampled negatives act only through BPR. The first control removes the label input and keeps signed negative regression, and the second removes signed negative regression.

\textbf{Search budgets.}
On ML and Food, each endpoint control and each fixed-support sampling control receives 150 trials, the BPR reference of Tab.~\ref{tab:endpoint_controls} and Figs.~\ref{fig:supervision_density} and~\ref{fig:component_ablation} receives a separate 200-trial search, and each component removal receives 500 trials. The ablation search space tunes $d\in\{32,64,128,256,512\}$, a log-uniform learning rate over $[10^{-5},10^{-1}]$, and log-uniform weight decay over $[10^{-8},10^{-3}]$, together with the coefficients of the endpoint controls. The flow-based controls also search the outer flow and trajectory weights over $[0.1,1]$ and $[0,0.5]$, and the sampling controls also tune $d_h\in\{64,128,256,512\}$ and $T\in\{1,4,8,16\}$. Each control search uses seed 2024, a training batch size of 2,048, an evaluation batch size of 4,096, and up to 500 epochs per trial, and it reports the test score of the configuration that validation NDCG@10 selects, independently of the full-model settings in Tab.~\ref{tab:best_hpo}. The PACF entries in these ablations repeat Tab.~\ref{tab:main}, and the BPR reference differs from the BPR-MF row of Tab.~\ref{tab:main}, which averages five seeds of the main search.

\textbf{Diagnostic sweeps.}
Tabs.~\ref{tab:tstep} and~\ref{tab:transport_width} vary $T$ and $d_h$ in fixed-budget diagnostic searches outside the main search. Their entries compare settings within each table, and Tab.~\ref{tab:main} reports the main result.

\textbf{Perturbations.}\label{sec:robustness_protocols}
Both perturbations change only the training interactions $\mathcal{O}_{\mathrm{train}}\subseteq\mathcal{O}$, so every setting evaluates the learned scorer on the same validation and test interactions and candidate sets, and every model keeps its clean-data configuration. Removal uses $\rho_s\in\{0,0.2,0.4,0.6,0.8\}$ on ML and $\rho_s\in\{0,0.2,0.4,0.6\}$ on Food. It first keeps one stratified interaction per user $u\in\mathcal{U}$, so every user retains a training history, then draws $N_{\mathrm{rem}}=\lfloor|\mathcal{O}_{\mathrm{train}}|(1-\rho_s)\rfloor-|\mathcal{U}|$ interactions uniformly without replacement. Replacement uses $\rho_n\in\{0,0.1,\ldots,0.5\}$ on both datasets. It samples $\mathcal{R}\subset\mathcal{O}_{\mathrm{train}}$ with $|\mathcal{R}|=\lfloor|\mathcal{O}_{\mathrm{train}}|\,\rho_n\rfloor$ and replaces the item of each $(u,i)\in\mathcal{R}$ with an item drawn uniformly from $\mathcal{I}\setminus\mathcal{I}_u$. The interaction count and density stay fixed while the observed associations change, and sampling among valid user-item edges keeps every perturbed set free of duplicate pairs.

\subsection{Full Performance Statistics}
\label{sec:full_performance_statistics}

\begin{table}[!t]
\centering
\caption{Comparison of all methods on ML, Food, Dance, and Yelp, measured by Recall@10 and NDCG@10 in percent as five-seed mean $\pm$ std. Rules separate discriminative baselines, generative baselines, and PACF, and \textbf{bold} and \underline{underline} mark the best and second-best mean per column. PACF has the highest mean in every column across the movie, short-video, and business domains.}
\label{tab:full_main}
\begingroup
\fontsize{9}{11}\selectfont
\setlength{\tabcolsep}{2pt}
\renewcommand{\arraystretch}{1.1}
\begin{tabular*}{\linewidth}{@{\extracolsep{\fill}}lcccccccc@{}}
\toprule
Method & \multicolumn{2}{c}{\textbf{ML}} & \multicolumn{2}{c}{\textbf{Food}} & \multicolumn{2}{c}{\textbf{Dance}} & \multicolumn{2}{c}{\textbf{Yelp}} \\
 & Recall & NDCG & Recall & NDCG & Recall & NDCG & Recall & NDCG \\ \midrule
BPR-MF & 5.59$\pm$1.34 & 2.66$\pm$0.65 & 6.91$\pm$0.22 & 3.36$\pm$0.08 & 6.62$\pm$0.24 & 3.29$\pm$0.08 & 2.77$\pm$0.26 & 1.37$\pm$0.13 \\
NeuMF & 6.18$\pm$0.52 & 3.11$\pm$0.15 & 7.38$\pm$0.32 & 3.55$\pm$0.13 & 6.66$\pm$0.24 & 3.31$\pm$0.07 & 3.11$\pm$0.13 & 1.54$\pm$0.04 \\
NCF & 6.39$\pm$0.44 & 3.24$\pm$0.11 & 7.57$\pm$0.30 & 3.63$\pm$0.11 & 6.72$\pm$0.19 & 3.37$\pm$0.11 & 3.19$\pm$0.11 & 1.60$\pm$0.05 \\
LightGCN & 6.52$\pm$0.24 & 3.16$\pm$0.11 & 7.44$\pm$0.34 & 3.60$\pm$0.15 & 6.19$\pm$0.30 & 3.09$\pm$0.13 & 3.26$\pm$0.06 & 1.62$\pm$0.04 \\
SGL & 6.72$\pm$0.18 & 3.35$\pm$0.06 & 9.08$\pm$0.03 & 4.48$\pm$0.02 & 9.23$\pm$0.14 & 4.83$\pm$0.04 & 4.23$\pm$0.07 & 2.10$\pm$0.02 \\
SimGCL & 6.13$\pm$1.05 & 2.98$\pm$0.56 & 9.46$\pm$0.35 & \underline{4.71$\pm$0.16} & 10.17$\pm$0.15 & 5.17$\pm$0.06 & 4.40$\pm$0.09 & 2.19$\pm$0.02 \\
LightCCF & \underline{7.38$\pm$0.62} & \underline{3.56$\pm$0.28} & \underline{9.56$\pm$0.13} & 4.68$\pm$0.06 & \underline{10.21$\pm$0.07} & \underline{5.18$\pm$0.04} & \underline{4.46$\pm$0.05} & \underline{2.21$\pm$0.01} \\ \midrule
MultiVAE & 6.27$\pm$0.42 & 3.09$\pm$0.13 & 8.36$\pm$0.02 & 4.09$\pm$0.00 & 8.37$\pm$0.24 & 4.24$\pm$0.12 & 3.54$\pm$0.08 & 1.77$\pm$0.04 \\
RecVAE & 6.82$\pm$0.25 & 3.14$\pm$0.15 & 8.02$\pm$0.13 & 3.93$\pm$0.08 & 8.74$\pm$0.36 & 4.56$\pm$0.15 & 3.62$\pm$0.22 & 1.81$\pm$0.08 \\
DiffRec & 6.82$\pm$0.56 & 3.54$\pm$0.15 & 7.34$\pm$0.33 & 3.53$\pm$0.13 & 8.65$\pm$2.10 & 4.26$\pm$1.02 & 3.08$\pm$0.08 & 1.53$\pm$0.03 \\
FlowCF & 7.31$\pm$0.22 & 3.45$\pm$0.08 & 8.36$\pm$0.11 & 4.31$\pm$0.10 & 9.05$\pm$0.15 & 4.79$\pm$0.06 & 3.56$\pm$1.57 & 1.72$\pm$0.77 \\ \midrule
PACF & \textbf{7.48$\pm$0.41} & \textbf{3.73$\pm$0.03} & \textbf{9.75$\pm$0.08} & \textbf{4.79$\pm$0.05} & \textbf{10.32$\pm$0.12} & \textbf{5.25$\pm$0.05} & \textbf{4.49$\pm$0.08} & \textbf{2.23$\pm$0.05} \\ \bottomrule
\end{tabular*}
\endgroup
\end{table}

\begin{table}[!t]
\centering
\caption{Comparison of all methods on Amazon Baby, Sports, and Clothing, measured by Recall@10 and NDCG@10 in percent as five-seed mean $\pm$ std. Rules separate discriminative baselines, generative baselines, and PACF, and \textbf{bold} and \underline{underline} mark the best and second-best mean per column. PACF has the highest mean in every column on the three sparsest datasets, all from e-commerce.}
\label{tab:full_main_amazon}
\begingroup
\fontsize{9}{11}\selectfont
\setlength{\tabcolsep}{2pt}
\renewcommand{\arraystretch}{1.1}
\begin{tabular}{@{}lcccccc@{}}
\toprule
Method & \multicolumn{2}{c}{\textbf{Baby}} & \multicolumn{2}{c}{\textbf{Sports}} & \multicolumn{2}{c}{\textbf{Clothing}} \\
 & Recall & NDCG & Recall & NDCG & Recall & NDCG \\ \midrule
BPR-MF & 1.42$\pm$0.04 & 0.67$\pm$0.01 & 1.48$\pm$0.05 & 0.74$\pm$0.03 & 0.64$\pm$0.15 & 0.30$\pm$0.08 \\
NeuMF & 1.43$\pm$0.04 & 0.68$\pm$0.01 & 1.49$\pm$0.05 & 0.74$\pm$0.03 & 0.65$\pm$0.17 & 0.31$\pm$0.09 \\
NCF & 1.68$\pm$0.03 & 0.82$\pm$0.01 & 1.48$\pm$0.06 & 0.75$\pm$0.03 & 0.68$\pm$0.18 & 0.32$\pm$0.09 \\
LightGCN & 1.37$\pm$0.04 & 0.66$\pm$0.01 & 1.34$\pm$0.13 & 0.67$\pm$0.06 & 0.58$\pm$0.09 & 0.27$\pm$0.04 \\
SGL & 2.16$\pm$0.09 & 1.11$\pm$0.04 & 2.29$\pm$0.24 & 1.19$\pm$0.11 & 1.30$\pm$0.26 & 0.69$\pm$0.15 \\
SimGCL & \underline{2.58$\pm$0.10} & 1.32$\pm$0.04 & 3.45$\pm$0.08 & 1.82$\pm$0.04 & 2.56$\pm$0.18 & 1.35$\pm$0.11 \\
LightCCF & 2.50$\pm$0.03 & \underline{1.40$\pm$0.03} & \underline{3.48$\pm$0.09} & \underline{1.91$\pm$0.04} & \underline{2.67$\pm$0.01} & \underline{1.47$\pm$0.02} \\ \midrule
MultiVAE & 1.67$\pm$0.04 & 0.80$\pm$0.01 & 1.36$\pm$0.12 & 0.69$\pm$0.06 & 1.49$\pm$0.05 & 0.77$\pm$0.04 \\
RecVAE & 1.61$\pm$0.16 & 0.80$\pm$0.08 & 1.60$\pm$0.56 & 0.81$\pm$0.29 & 1.47$\pm$0.11 & 0.79$\pm$0.06 \\
DiffRec & 1.57$\pm$0.02 & 0.72$\pm$0.00 & 1.58$\pm$0.05 & 0.85$\pm$0.00 & 0.65$\pm$0.03 & 0.30$\pm$0.01 \\
FlowCF & 1.79$\pm$0.02 & 0.91$\pm$0.02 & 1.98$\pm$0.01 & 1.08$\pm$0.00 & 1.22$\pm$0.02 & 0.66$\pm$0.00 \\ \midrule
PACF & \textbf{2.63$\pm$0.04} & \textbf{1.42$\pm$0.03} & \textbf{3.69$\pm$0.07} & \textbf{1.97$\pm$0.03} & \textbf{2.88$\pm$0.04} & \textbf{1.59$\pm$0.02} \\ \bottomrule
\end{tabular}
\endgroup
\end{table}

Tabs.~\ref{tab:full_main} and~\ref{tab:full_main_amazon} report the five-seed runs of Tab.~\ref{tab:main} with their standard deviations. PACF has the highest mean in every column, and its seed standard deviation stays below its margin over BPR-MF in every column. PACF and BPR-MF share the inner-product scorer of Eq.~\eqref{eq:pred}, so this margin comes entirely from the trained embeddings. The following sections trace how the pathwise losses produce it.

\subsection{How the Shared Field Updates the Embeddings}\label{sec:embedding_diagnostic}\label{sec:visualizing}

Gradient amortization needs two effects. The field must learn the BCF targets, and its residuals must move the embeddings that the scorer keeps. Fig.~\ref{fig:r5_transport} examines the first effect and Tab.~\ref{tab:gradient_drift} the second.

Fig.~\ref{fig:r5_transport} trains the field on two-dimensional synthetic data. With flow matching, positive predictions follow the displacement $\mathbf{q}_i-\mathbf{p}_u$ and negative predictions reverse it, as Eq.~\eqref{eq:target_vector} requires. Without $\mathcal{L}_{\mathrm{flow}}$, the trajectory term supervises only positive-label inputs, and the field points toward the item under both labels. Flow matching is the only loss shaping the field at negative-label inputs.

\begin{table}[!t]
\centering
\caption{Gradient and embedding-drift ratios of $\mathcal{L}_{\mathrm{flow}}$ to $\mathcal{L}_{\mathrm{rec}}$ at PACF checkpoints on ML and Food. Flow matching moves PACF's embeddings by amounts comparable to the ranking loss, and its input route alone carries a nonzero gradient, so it directly reshapes the scorer used at inference.}
\label{tab:gradient_drift}
\begin{tabular}{lccc}
\toprule
Dataset &
$\frac{\|\nabla_{\mathbf{P},\mathbf{Q}}\mathcal{L}_{\mathrm{flow}}\|_2}
{\|\nabla_{\mathbf{P},\mathbf{Q}}\mathcal{L}_{\mathrm{rec}}\|_2}$ &
$\frac{\Delta_{\mathbf{P},\mathbf{Q}}(\mathcal{L}_{\mathrm{flow}})}
{\Delta_{\mathbf{P},\mathbf{Q}}(\mathcal{L}_{\mathrm{rec}})}$ &
$\frac{\|\nabla_{\mathbf{P},\mathbf{Q}}\mathcal{L}_{\mathrm{flow}}^{\mathrm{detach}}\|_2}
{\|\nabla_{\mathbf{P},\mathbf{Q}}\mathcal{L}_{\mathrm{rec}}\|_2}$ \\
\midrule
ML   & 0.1375 & 0.9915 & 0.1095 \\
Food & 0.2399 & 1.4893 & 0.0716 \\
\bottomrule
\end{tabular}
\end{table}

Tab.~\ref{tab:gradient_drift} measures how strongly this fit reaches the embedding tables at PACF checkpoints. It compares gradient norms on $(\mathbf{P},\mathbf{Q})$ with the one-step displacement $\Delta_{\mathbf{P},\mathbf{Q}}(\mathcal{L})$, the $\ell_2$ norm of the change in $(\mathbf{P},\mathbf{Q})$ after one Adam step on a loss $\mathcal{L}$. The control $\mathcal{L}_{\mathrm{flow}}^{\mathrm{detach}}$ stops the gradient through the target $\mathbf{u}_t$ and keeps the error, so its gradient flows only through the state $\mathbf{h}_t$ and the item condition. The first column compares raw gradient magnitudes, the second measures the movement that Adam applies to the embeddings, and the third keeps only the part of the flow gradient that bypasses the target.

\begin{wrapfigure}[18]{r}{0.5\columnwidth}
  \centering
  \includegraphics[width=.75\linewidth]{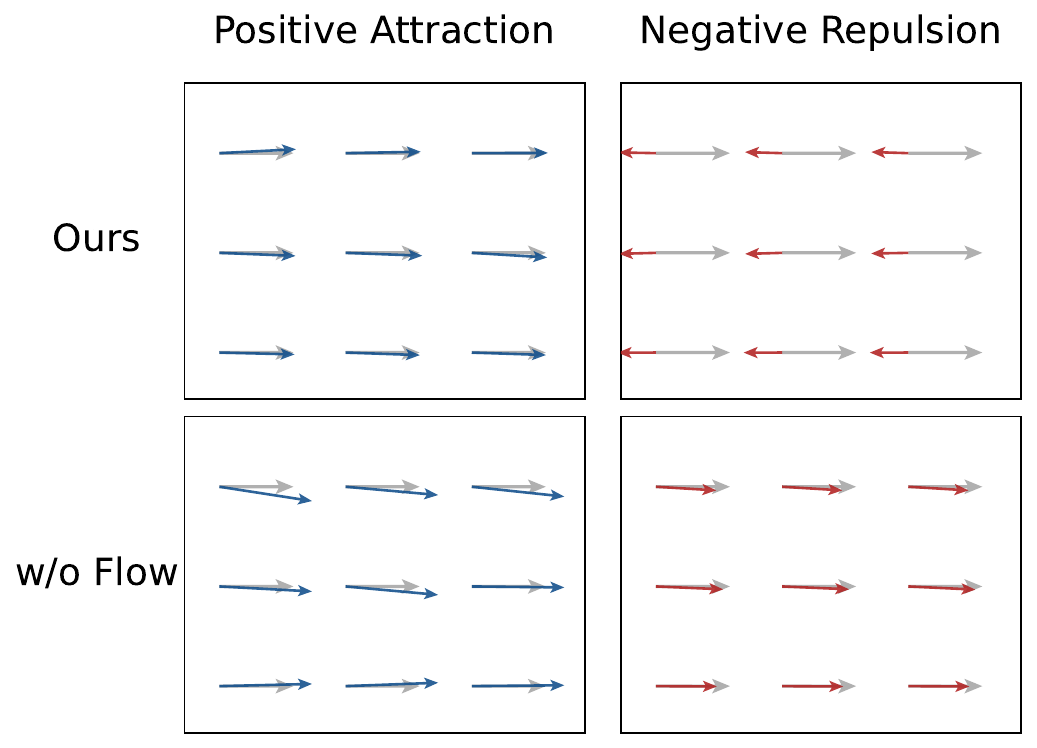}
  \caption{Field predictions on synthetic data, with one column per label and one row each for PACF with and without $\mathcal{L}_{\mathrm{flow}}$. Gray arrows show $\mathbf{q}_i-\mathbf{p}_u$, and colored arrows show $\mathbf{v}_{\boldsymbol\theta}$. Flow matching makes PACF's field point toward the item for positive labels and away from it for negative labels.}
  \label{fig:r5_transport}
\end{wrapfigure}

The flow gradient has a smaller norm than the ranking gradient, yet one Adam step on $\mathcal{L}_{\mathrm{flow}}$ moves the embeddings about as far as a step on $\mathcal{L}_{\mathrm{rec}}$ on ML and farther on Food. Adam divides each coordinate by its second-moment estimate, so the displacement follows how consistently a loss pushes each coordinate. The detached gradient is nonzero on both datasets, so by Eq.~\eqref{eq:flow_grad_residual} flow residuals reach $\mathbf{p}_u$ through the input route $(1-t)\mathbf{J}_{\mathbf h}$ even with the target held fixed. Restoring the target route raises the gradient norm on both datasets. This route carries the item-end pull of each user toward the mean of its co-users. It acts on the same residual as the input route, so their relative directions set the size of the gradient. Flow matching therefore changes the embeddings at a rate comparable to the ranking loss, and the inner-product scorer keeps this change after the shared velocity field is discarded.

\subsection{Contributions of the Pathwise Components}\label{sec:ablation_variants}\label{sec:ablation_recall_curves}

The main text measures each choice by NDCG@10. This section reports these comparisons by Recall@10, adds loss controls and integration depth, and explains each effect through the regression.

\begin{figure}[!t]
  \centering
  \includegraphics[width=0.58\linewidth]{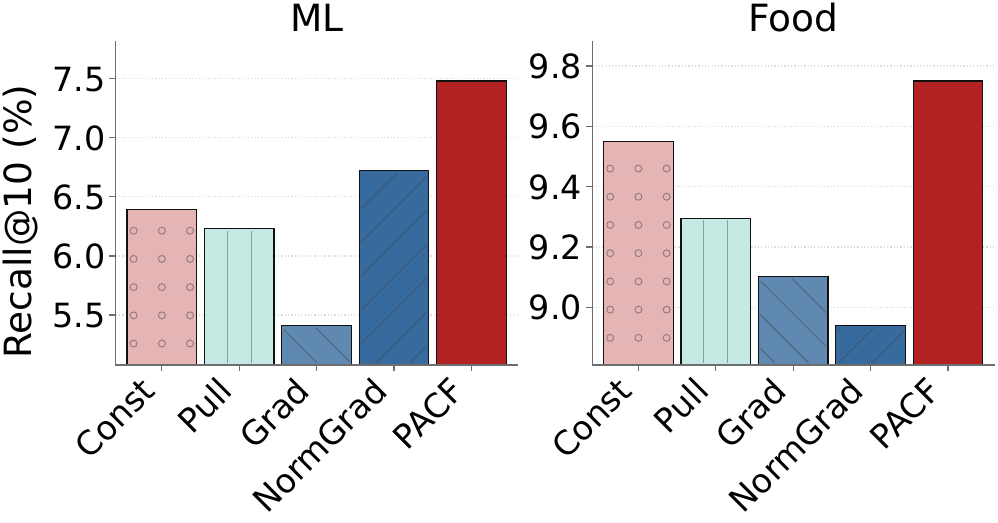}
  \caption{Recall@10 in percent for the target comparison of Fig.~\ref{fig:target_structure_ndcg}. Among the target controls, NormGrad leads on ML and Const on Food. PACF reports the separately tuned main result.}
  \label{fig:target_structure_recall}
\end{figure}

\textbf{Target form.}
Fig.~\ref{fig:target_structure_recall} adds Recall to the NDCG comparison of Sec.~\ref{sec:ablation_target}. Grad is $y_{ui}\mathbf{q}_i$, the gradient of the signed inner-product score with respect to the user embedding, and NormGrad rescales it to unit norm. The four targets share one tuning protocol, and the PACF bar repeats the main result.

NormGrad gives the highest control Recall on ML, while Const leads on Food. The different ML rankings under Recall and NDCG distinguish retrieval within the top ten from the higher positions that NDCG rewards. On Food, Const leads both metrics. Together, these comparisons evaluate how the target's user dependence, magnitude, and sign affect the static ranking learned through the shared field. Const retains the full user-to-item displacement and reverses it for sampled negatives, preserving both user differences and the distinction between the two labels throughout the path.

\begin{figure}[!t]
\centering
\begin{minipage}[t]{0.49\linewidth}
  \centering
  \includegraphics[width=0.87\linewidth]{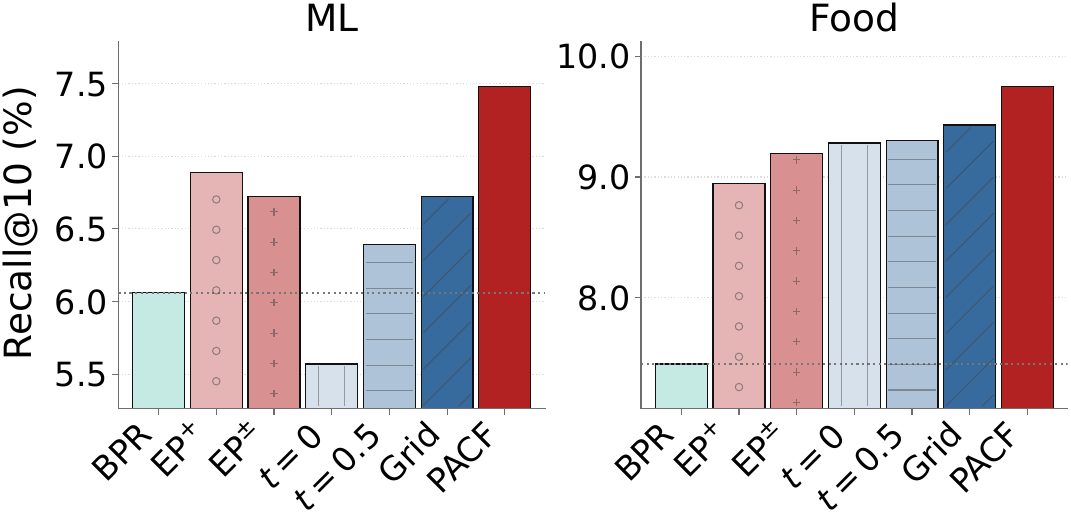}
  \caption{Sampling supports of $t$ in $\mathcal{L}_{\mathrm{flow}}$ on ML and Food, shown with BPR and the two endpoint penalties and measured by Recall@10 in percent. The dotted line marks the BPR level. PACF samples $t$ uniformly over the whole segment.}
  \label{fig:supervision_density_recall}
\end{minipage}\hfill
\begin{minipage}[t]{0.49\linewidth}
  \centering
  \includegraphics[width=\linewidth]{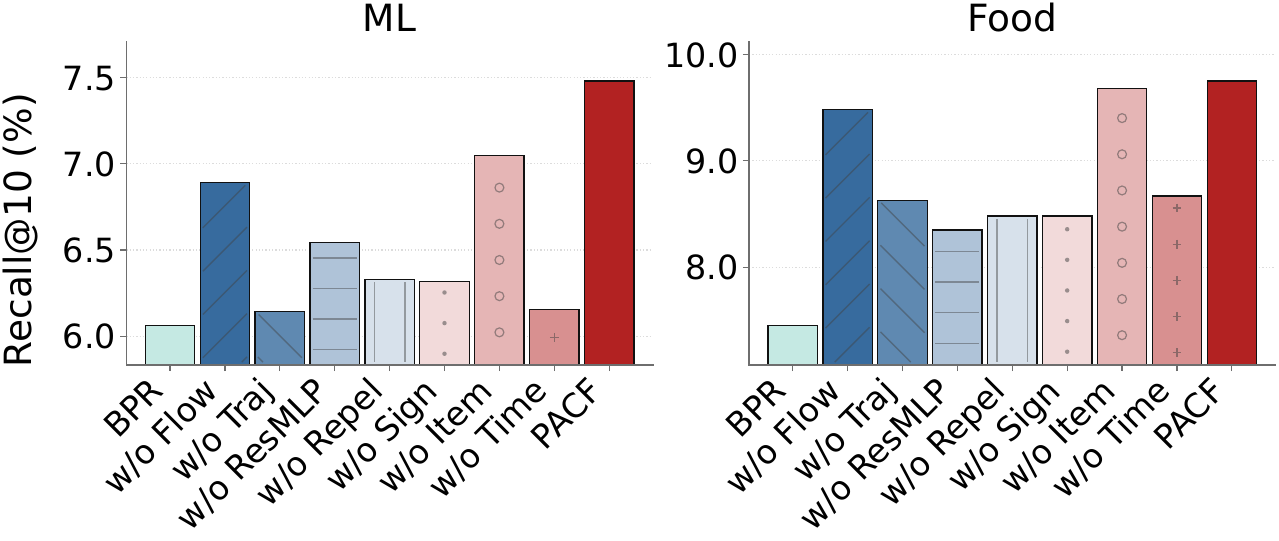}
  \caption{Removal of one PACF component at a time on ML and Food, measured by Recall@10 in percent against BPR. PACF feeds the complete input $(\mathbf{h}_t,\mathbf{q}_i,t,y_{ui})$ to one shared residual field and trains it with both of its pathwise loss terms.}
  \label{fig:component_ablation_recall}
\end{minipage}
\end{figure}

\textbf{Sampling support.}
Fig.~\ref{fig:supervision_density_recall} repeats the sampling supports of Fig.~\ref{fig:supervision_density} under Recall@10, and Recall rises from $t{=}0$ through $t{=}0.5$ to Grid on both datasets. At $t=0$ every state equals its user embedding, so each user presents its own input and the field fits each displacement separately. Toward the item, the states of co-users approach one another by the factor $1-t$ while their targets keep their difference, and the coefficient $(|\mu(y)|-L_{\mathbf h}(1-t))_+$ in Eq.~\eqref{eq:item_end_pointwise} grows. Uniform sampling covers the item-end interval where this coefficient is positive and also spreads the states over the whole segment, so its comparison with Grid changes the coverage of the segment and the proximity to the item together.

\textbf{Field inputs.}
Fig.~\ref{fig:component_ablation_recall} repeats the component removals of Fig.~\ref{fig:component_ablation} under Recall@10. The complete input determines the target through Eq.~\eqref{eq:target_function}, which needs the state, the item, the time, and the label. Removing the item condition forces the field to infer $\mathbf{q}_i$ from a state that mixes user and item. Removing time leaves the remaining displacement $\mathbf{q}_i-\mathbf{h}_t$ without its scale $1/(1-t)$, so one input fits a nearby user early on the segment and a distant user late on it. In both cases the residual records a missing input instead of a difference between users. Removing time lowers Recall more than removing the item condition on both datasets. The state still carries part of the item through $t\mathbf{q}_i$, while only $t$ supplies the scale. The BCF target is linear in the embeddings up to the sign $\mu(y)$, so the residual skip of Eq.~\eqref{eq:resmlp} carries the linear displacement and the nonlinear branch absorbs the sign and the conditional-mean correction. An equally sized plain MLP represents both through its nonlinear layers, which changes its fit and the state sensitivity $\mathbf{J}_{\mathbf h}$ that carries residuals into the embeddings.

\textbf{Labels and negative targets.}
Zeroing the label in \textbf{w/o Sign} gives a positive and a negative pair near the same item similar inputs with opposite targets, so the field averages them and the residuals mix the two labels. \textbf{w/o Repel} removes negative regression and leaves negatives to BPR. Tab.~\ref{tab:trajectory_controls} adds a finer control. Setting $\mu_-=0$ keeps the negative flow terms and turns each into $\|\mathbf{v}_{\boldsymbol\theta}\|_2^2$, which removes the displacement target and the gradient through it, and every mean falls below PACF. At $\mu_-=0$ the negative-label bound of Prop.~\ref{prop:item_end} vanishes, so the signed negative targets of PACF add an item-end condition on the users who sample each item as a negative, with the weight $\mu_-^2\leq1$.

\begin{table}[!t]
\centering
\caption{Trajectory-term controls on ML and Food, measured by Recall@10 and NDCG@10 in percent as five-seed mean $\pm$ std. PACF exceeds every control in all four columns, including the Beta sampling control, so both pathwise terms and the signed negative target contribute to its ranking.}
\label{tab:trajectory_controls}
\resizebox{\linewidth}{!}{
\begin{tabular}{lcccc}
\toprule
& \multicolumn{2}{c}{ML} & \multicolumn{2}{c}{Food} \\
Setting & Recall & NDCG & Recall & NDCG \\ \midrule
PACF & $7.48\pm0.41$ & $3.73\pm0.03$ & $9.75\pm0.08$ & $4.79\pm0.05$ \\
w/o $\mathcal{L}_{\mathrm{traj}}$ & $7.05\pm0.44$ & $3.33\pm0.13$ & $9.32\pm0.16$ & $4.57\pm0.06$ \\
$\beta=0$, $t\sim\Beta(1,b)$ & $7.25\pm0.47$ & $3.64\pm0.27$ & $9.21\pm0.19$ & $4.54\pm0.09$ \\
w/o $\mathcal{L}_{\mathrm{flow}}$ & $6.69\pm0.76$ & $3.22\pm0.42$ & $9.27\pm0.08$ & $4.52\pm0.06$ \\
$\mu_-=0$ & $7.08\pm0.66$ & $3.33\pm0.35$ & $9.17\pm0.17$ & $4.50\pm0.12$ \\ \bottomrule
\end{tabular}}
\end{table}

\textbf{Flow and trajectory losses.}
Tab.~\ref{tab:trajectory_controls} retunes each control over five seeds, matching the protocol of the main result. The control w/o $\mathcal{L}_{\mathrm{traj}}$ sets $\beta=0$, w/o $\mathcal{L}_{\mathrm{flow}}$ sets $\lambda=0$, and the Beta control sets $\beta=0$ and samples $t\sim\Beta(1,b)$ with $b>1$, which moves the flow-matching states toward the user. The two losses constrain different quantities. Flow matching fits the field at reference states, and trajectory consistency checks the endpoint of the RK4 stages, where local errors add up or cancel. Removing either loss lowers both metrics on both datasets, and each drop exceeds the seed standard deviation on Food and for NDCG on ML. The Beta control recovers part of the NDCG decline on ML, stays at the level of w/o $\mathcal{L}_{\mathrm{traj}}$ on Food, and remains below PACF on both datasets. The trajectory term therefore contributes beyond the placement of the flow states. Setting $\beta=0$ removes the final RK4 stage at $t=1$, which Remark~\ref{rem:traj_endpoint} shows to couple all users of an item through one field output.

\begin{wraptable}[12]{r}{0.5\columnwidth}
\centering
\caption{Number of RK4 steps $T$ on ML and Food, measured by Recall@10 and NDCG@10 in percent. One RK4 step comes close to the deeper settings, so PACF can keep integration short.}
\label{tab:tstep}
\resizebox{\linewidth}{!}{
\begin{tabular}{llcccc}
\toprule
Dataset               & Metric & $T=1$ & $T=4$ & $T=8$ & $T=16$ \\ \midrule
\multirow{2}{*}{ML}   & Recall & 6.39  & 6.39  & 6.39  & 6.39   \\
                      & NDCG   & 3.16  & 3.05  & 3.21  & 3.11   \\ \midrule
\multirow{2}{*}{Food} & Recall & 9.20  & 9.13  & 9.17  & 8.55   \\
                      & NDCG   & 4.59  & 4.56  & 4.53  & 4.33   \\ \bottomrule
\end{tabular}
}
\end{wraptable}

\textbf{Integration depth.}
Tab.~\ref{tab:tstep} varies the number of RK4 steps $T$. ML stays nearly constant across $T$, and Food stays level up to $T=8$ and declines at $T=16$. Two properties of the solver make short integration sufficient. Each RK4 step evaluates the field at four generated states, so $T=1$ already constrains the endpoint through composed predictions. The final-stage coupling of Remark~\ref{rem:traj_endpoint} carries the weight $(6T)^{-2}$, so deeper integration dilutes the item-end constraint of the trajectory term, while training cost grows with the $4TB$ trajectory evaluations of the complexity analysis. Every listed dataset in Tab.~\ref{tab:best_hpo} accordingly selects at most four RK4 steps.

\subsection{Field Capacity and Static Ranking}\label{sec:velocity_width}

PACF uses the field as a regressor whose residuals train the embeddings, so fit quality and ranking quality can separate. Tab.~\ref{tab:transport_width} measures both as the width $d_h$ grows. Writing $\Phi_1(\mathbf{p}\mid\mathbf{q},y)$ for the RK4 endpoint from $\mathbf{p}$ under the condition $(\mathbf{q},y)$, we normalize the endpoint error by the initial distance:
\begin{equation}
\rho^+_u=\frac{\|\Phi_1(\mathbf{p}_u\mid\mathbf{q}_{i^+},+1)-\mathbf{q}_{i^+}\|_2}{\|\mathbf{q}_{i^+}-\mathbf{p}_u\|_2},
\qquad
\rho^-_u=\frac{\|\Phi_1(\mathbf{p}_u\mid\mathbf{q}_{i^-},+1)-\mathbf{q}_{i^-}\|_2}{\|\mathbf{q}_{i^-}-\mathbf{p}_u\|_2}.
\label{eq:transport_residual_ratios}
\end{equation}
Both ratios integrate the positive-label trajectory. The ratio $\rho^+_u$ measures arrival at the observed item, and $\rho^-_u$ measures how far the same trajectory carries the user toward a sampled negative item. Each entry averages the per-checkpoint median over all evaluated users across the five seeds.

\begin{table}[!t]
\centering
\setlength{\tabcolsep}{2pt}
\caption{Hidden width $d_h$ of the velocity network on ML and Food, with residual ratios and Recall@10 and NDCG@10 in percent as five-seed mean $\pm$ std. Endpoint fit tightens as the field widens while ranking changes non-monotonically, because PACF ranks with static embeddings.}
\label{tab:transport_width}
\resizebox{\linewidth}{!}{
\begin{tabular}{lcccccccc}
\toprule
& \multicolumn{4}{c}{ML} & \multicolumn{4}{c}{Food} \\
$d_h$ & $\rho^+$ & $\rho^-$ & Recall & NDCG
& $\rho^+$ & $\rho^-$ & Recall & NDCG \\ \midrule
$64$ & $0.5363\pm0.0699$ & $0.4833\pm0.0706$ & $6.36\pm0.66$ & $3.09\pm0.29$
& $0.3116\pm0.0162$ & $1.1969\pm0.3042$ & $8.99\pm0.22$ & $4.39\pm0.12$ \\
$128$ & $0.2554\pm0.0512$ & $0.2474\pm0.0471$ & $6.10\pm1.10$ & $2.84\pm0.47$
& $0.2457\pm0.0047$ & $2.0545\pm0.2827$ & $9.16\pm0.19$ & $4.47\pm0.08$ \\
$256$ & $0.1624\pm0.0012$ & $0.1883\pm0.0069$ & $6.56\pm0.79$ & $3.04\pm0.37$
& $0.2042\pm0.0094$ & $1.2037\pm0.3624$ & $9.35\pm0.09$ & $4.59\pm0.06$ \\
$512$ & $0.1390\pm0.0066$ & $0.1755\pm0.0109$ & $6.59\pm0.73$ & $3.18\pm0.32$
& $0.1461\pm0.0168$ & $0.2491\pm0.0318$ & $9.19\pm0.05$ & $4.48\pm0.04$ \\ \bottomrule
\end{tabular}}
\end{table}

The endpoint fit improves monotonically with width on both datasets, while ranking changes non-monotonically, and on Food a narrower field ranks better than the widest one at a larger $\rho^+$. Eq.~\eqref{eq:flow_grad_residual} explains this separation. Embedding updates scale with the flow residual, and Prop.~\ref{prop:item_end} bounds the residuals of co-users by a coefficient that falls as the state sensitivity $L_{\mathbf h}$ grows. A more expressive field fits the displacements with smaller residuals and sends smaller updates into the embeddings, so a tighter endpoint fit can weaken the signal that reaches the scorer. The widest field also brings users close to sampled negative items under a supplied positive label, while $\rho^-$ exceeds one on Food at narrower widths. An endpoint readout then finds every candidate close to its own endpoint, which explains the weak endpoint readouts in Tab.~\ref{tab:inference_readout}. Endpoint residuals therefore diagnose the conditional regression, and ranking metrics evaluate the static inner-product scorer that PACF deploys.

\subsection{Scoring and Path Geometry}

\begin{table}[!t]
\centering
\caption{Inner-product and cosine scoring under straight and spherical-geodesic reference paths on ML and Food, measured by Recall@10 and NDCG@10 in percent as five-seed mean $\pm$ std. The variants score closely, so the pathwise losses of PACF carry over to cosine scoring and geodesic paths.}
\label{tab:cosine_geometry}
\setlength{\aboverulesep}{0pt}
\setlength{\belowrulesep}{0pt}
\renewcommand{\arraystretch}{1.2}
\resizebox{\linewidth}{!}{
\begin{tabular}{ll|cccc}
\toprule
\multicolumn{2}{c|}{Variant} & \multicolumn{2}{c}{ML} & \multicolumn{2}{c}{Food} \\
Score & Path & Recall & NDCG & Recall & NDCG \\ \midrule
Inner product & Straight & $7.48\pm0.41$ & $3.73\pm0.03$ & $9.75\pm0.08$ & $4.79\pm0.05$ \\
Cosine & Straight & $7.74\pm0.96$ & $3.71\pm0.50$ & $10.01\pm0.09$ & $4.93\pm0.03$ \\
Cosine & Geodesic & $7.18\pm0.48$ & $3.42\pm0.23$ & $9.93\pm0.27$ & $4.91\pm0.08$ \\ \bottomrule
\end{tabular}}
\end{table}

\begin{table}[!t]
\centering
\caption{Learned item geometry under each scorer and reference path on ML and Food. With the inner product and straight paths, PACF keeps item norms widely spread and closely tied to popularity.}
\label{tab:cosine_geometry_diagnostics}
\setlength{\aboverulesep}{0pt}
\setlength{\belowrulesep}{0pt}
\renewcommand{\arraystretch}{1.2}
\resizebox{\linewidth}{!}{
\begin{tabular}{ll|cccccc}
\toprule
\multicolumn{2}{c|}{Variant} & \multicolumn{3}{c}{ML} & \multicolumn{3}{c}{Food} \\
Score & Path & Norm std. & Correlation & Effective rank & Norm std. & Correlation & Effective rank \\ \midrule
Inner product & Straight & $0.1890$ & $+0.6900$ & $311.8$ & $0.8892$ & $+0.9967$ & $371.6$ \\
Cosine & Straight & $0.0780$ & $+0.6342$ & $111.9$ & $0.0322$ & $+0.9232$ & $394.0$ \\
Cosine & Geodesic & $0.0111$ & $+0.2513$ & $101.4$ & $0.0153$ & $-0.0151$ & $400.7$ \\ \bottomrule
\end{tabular}}
\end{table}

PACF regresses along straight paths in $\mathbb{R}^d$ and ranks with inner products, so embedding norms enter both the pathwise targets and the score. Tab.~\ref{tab:cosine_geometry} replaces the scorer with cosine similarity and then the path with a spherical geodesic. Tab.~\ref{tab:cosine_geometry_diagnostics} measures item geometry by the standard deviation of item norms, their correlation with popularity, and the effective rank of the normalized spectrum.

Cosine scoring with straight paths raises the mean on Food and on ML Recall, and the inner product keeps the highest ML NDCG, so the pathwise losses carry over to a cosine scorer. Straight paths keep the embedding norms in their states and targets, so the regression uses radial information while a cosine scorer orders items by angle. Geodesics between normalized directions remove this radial information. Item norms collapse to nearly one value, their correlation with popularity weakens on ML and vanishes on Food, and cosine scoring with geodesics falls below cosine scoring with straight paths in every column. With the inner product and straight paths, item norms stay spread and track popularity, so the MF scorer keeps a popularity component that cosine scoring discards.

\subsection{Robustness to Sparsity and Noise}

\begin{figure}[!t]
\centering
\begin{subfigure}[b]{0.495\linewidth}
  \includegraphics[width=\linewidth]{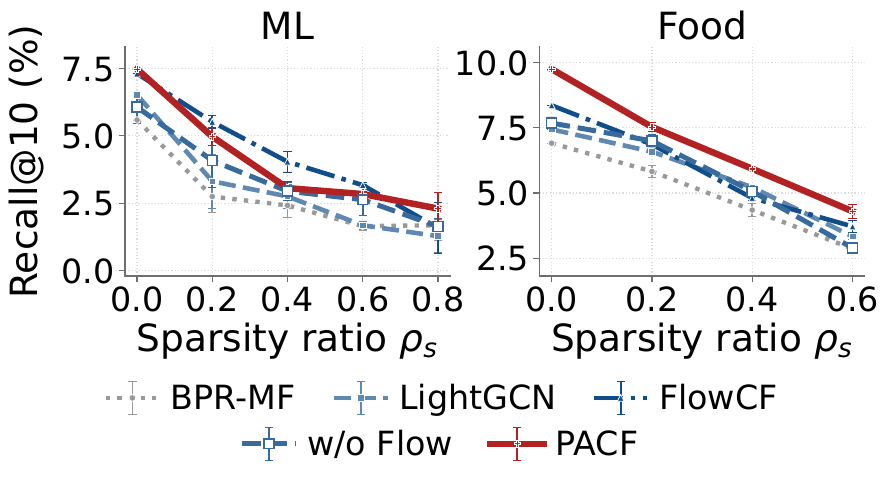}
  \caption{Data sparsity.}
  \label{fig:robustness_sparsity_recall}
\end{subfigure}\hfill
\begin{subfigure}[b]{0.495\linewidth}
  \includegraphics[width=\linewidth]{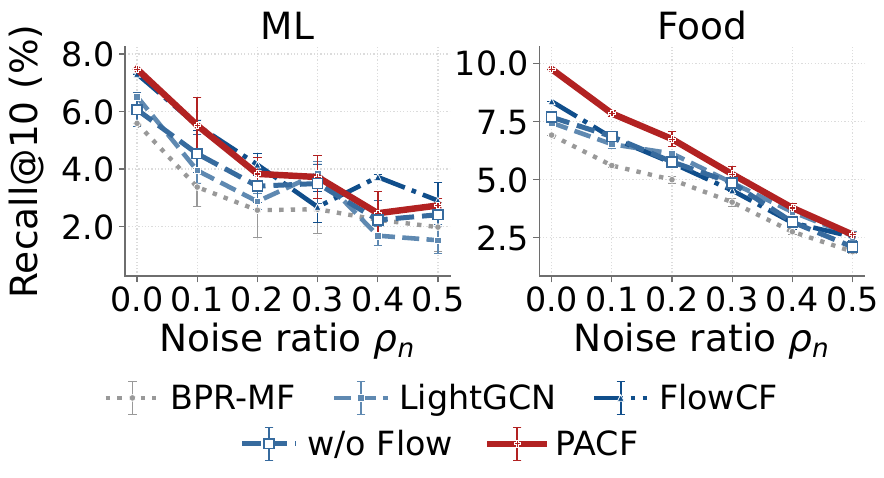}
  \caption{Replacement noise.}
  \label{fig:robustness_noise_recall}
\end{subfigure}
\caption{Robustness of PACF and baselines to interaction removal in panel a and item replacement in panel b on ML and Food, measured by Recall@10 in percent. Error bars show seed standard deviations. PACF stays above w/o Flow at every level and leads all compared models on Food.}
\label{fig:robustness_recall}
\end{figure}

Fig.~\ref{fig:robustness_recall} repeats the perturbation study of Fig.~\ref{fig:robustness_main} under Recall@10. On ML, FlowCF leads at intermediate removal levels, and under replacement LightGCN leads at one level while FlowCF leads or ties at the others. PACF stays above w/o Flow at every level of both perturbations on both datasets and has the highest Recall at every level on Food, the sparser dataset with fewer interactions per user.

Removal reduces the observed pairs of each user, and Eq.~\eqref{eq:ranking_loss} constrains a user embedding only through these pairs. Every retained pair still enters the shared regression, so the item-end bound of Prop.~\ref{prop:item_end} keeps acting on the users who share each remaining item, and the gap to w/o Flow measures this contribution. Replacement keeps the number of pairs and corrupts their associations, so part of the flow targets point to random items. PACF stays above w/o Flow at every replacement level, so the shared regression keeps its benefit when some targets are wrong. ML is the densest dataset, where each user keeps the most constraints of its own, and it also shows the smallest gain of PACF over BPR-MF in the main comparison, which matches its weaker lead under perturbation.

\subsection{Computational Cost}\label{sec:efficiency}

\begin{table}[!t]
\centering
\caption{Computational cost of nine methods on six datasets, with training time in seconds per epoch, inference latency in milliseconds per user, and peak GPU memory in GB. PACF keeps the latency of an inner-product scorer and trains fastest on average after BPR-MF on Sports and Clothing.}
\label{tab:efficiency}
\resizebox{\linewidth}{!}{
\begin{tabular}{clccccccccc}
\toprule
\multicolumn{1}{l}{Dataset} & Metric          & BPR-MF            & NCF               & NeuMF             & LightGCN          & MultiVAE          & RecVAE            & DiffRec           & FlowCF            & PACF              \\ \midrule
\multirow{3}{*}{ML}         & Training   & $0.96 \pm 1.13$   & $4.31 \pm 0.30$   & $3.98 \pm 0.34$   & $1.29 \pm 1.16$   & $1.04 \pm 0.18$   & $1.10 \pm 0.21$   & $1.09 \pm 0.19$   & $7.12 \pm 1.21$   & $2.33 \pm 0.20$   \\
                            & Inference  & $0.02 \pm 0.00$   & $0.26 \pm 0.01$   & $0.38 \pm 0.00$   & $0.03 \pm 0.00$   & $0.13 \pm 0.00$   & $0.46 \pm 0.07$   & $25.36 \pm 0.13$  & $0.68 \pm 0.05$   & $0.07 \pm 0.00$   \\
                            & GPU      & $<0.01 \pm <0.01$ & $0.10 \pm 0.00$   & $0.10 \pm 0.00$   & $<0.01 \pm <0.01$ & $0.30 \pm 0.00$   & $0.40 \pm 0.00$   & $0.40 \pm 0.00$   & $0.30 \pm 0.00$   & $1.10 \pm 0.00$   \\ \midrule
\multirow{3}{*}{Food}       & Training   & $0.39 \pm 0.46$   & $1.37 \pm 0.37$   & $1.35 \pm 0.37$   & $0.46 \pm 0.54$   & $0.46 \pm 0.69$   & $0.51 \pm 0.63$   & $0.38 \pm 0.43$   & $1.50 \pm 0.49$   & $0.86 \pm 0.67$   \\
                            & Inference  & $0.03 \pm 0.00$   & $0.23 \pm 0.05$   & $0.33 \pm 0.01$   & $0.03 \pm 0.00$   & $0.14 \pm 0.01$   & $0.45 \pm 0.01$   & $27.17 \pm 0.12$  & $0.68 \pm 0.01$   & $0.07 \pm 0.00$   \\
                            & GPU      & $<0.01 \pm <0.01$ & $0.10 \pm 0.00$   & $0.10 \pm 0.00$   & $<0.01 \pm <0.01$ & $0.20 \pm 0.00$   & $0.30 \pm 0.00$   & $0.20 \pm 0.00$   & $0.10 \pm 0.00$   & $0.90 \pm 0.00$   \\ \midrule
\multirow{3}{*}{Dance}      & Training   & $0.64 \pm 0.74$   & $2.48 \pm 0.63$   & $2.94 \pm 0.72$   & $0.91 \pm 0.89$   & $0.91 \pm 0.85$   & $0.81 \pm 0.78$   & $0.74 \pm 0.74$   & $4.97 \pm 0.93$   & $1.96 \pm 0.90$   \\
                            & Inference  & $0.03 \pm 0.00$   & $0.23 \pm 0.01$   & $0.34 \pm 0.03$   & $0.03 \pm 0.00$   & $0.14 \pm 0.00$   & $0.45 \pm 0.01$   & $25.78 \pm 0.11$  & $0.68 \pm 0.03$   & $0.07 \pm 0.00$   \\
                            & GPU      & $<0.01 \pm <0.01$ & $0.10 \pm 0.00$   & $0.10 \pm 0.00$   & $0.10 \pm 0.00$   & $0.20 \pm 0.00$   & $0.30 \pm 0.00$   & $0.30 \pm 0.00$   & $0.20 \pm 0.00$   & $0.90 \pm 0.00$   \\ \midrule
\multirow{3}{*}{Baby}       & Training   & $2.87 \pm 4.31$   & $7.40 \pm 4.17$   & $6.25 \pm 3.56$   & $2.69 \pm 4.05$   & $3.37 \pm 4.29$   & $3.88 \pm 4.44$   & $3.85 \pm 4.31$   & $8.70 \pm 3.78$   & $4.67 \pm 4.19$   \\
                            & Inference  & $0.03 \pm 0.00$   & $0.30 \pm 0.00$   & $0.46 \pm 0.00$   & $0.03 \pm 0.00$   & $0.14 \pm 0.00$   & $0.46 \pm 0.01$   & $45.33 \pm 0.11$  & $0.66 \pm 0.01$   & $0.07 \pm 0.00$   \\
                            & GPU      & $<0.01 \pm <0.01$ & $0.10 \pm 0.00$   & $0.10 \pm 0.00$   & $0.10 \pm 0.00$   & $1.00 \pm 0.00$   & $1.00 \pm 0.00$   & $1.10 \pm 0.00$   & $0.50 \pm 0.00$   & $0.90 \pm 0.00$   \\ \midrule
\multirow{3}{*}{Sports}     & Training   & $10.20 \pm 19.00$ & $18.56 \pm 18.59$ & $19.84 \pm 19.90$ & $11.75 \pm 19.55$ & $13.58 \pm 19.23$ & $11.69 \pm 15.66$ & $13.19 \pm 18.33$ & $28.40 \pm 19.68$ & $11.13 \pm 15.66$ \\
                            & Inference  & $0.03 \pm 0.00$   & $0.47 \pm 0.01$   & $0.74 \pm 0.00$   & $0.03 \pm 0.00$   & $0.28 \pm 0.00$   & $0.45 \pm 0.01$   & $81.37 \pm 0.09$  & $0.95 \pm 0.07$   & $0.07 \pm 0.00$   \\
                            & GPU      & $0.10 \pm 0.00$   & $0.10 \pm 0.00$   & $0.20 \pm 0.00$   & $0.30 \pm 0.00$   & $3.60 \pm 0.00$   & $3.60 \pm 0.00$   & $4.10 \pm 0.00$   & $1.20 \pm 0.00$   & $1.10 \pm 0.00$   \\ \midrule
\multirow{3}{*}{Clothing}   & Training   & $10.99 \pm 22.05$ & $20.81 \pm 26.45$ & $21.96 \pm 28.02$ & $14.28 \pm 26.66$ & $17.53 \pm 28.40$ & $14.72 \pm 22.00$ & $17.13 \pm 26.60$ & $31.11 \pm 28.31$ & $13.98 \pm 22.75$ \\
                            & Inference  & $0.03 \pm 0.00$   & $0.59 \pm 0.01$   & $0.87 \pm 0.00$   & $0.04 \pm 0.00$   & $0.52 \pm 0.00$   & $0.45 \pm 0.00$   & $96.84 \pm 0.11$  & $0.82 \pm 0.09$   & $0.07 \pm 0.00$   \\
                            & GPU      & $0.10 \pm 0.00$   & $0.10 \pm 0.00$   & $0.20 \pm 0.00$   & $0.30 \pm 0.00$   & $4.80 \pm 0.00$   & $4.80 \pm 0.00$   & $5.50 \pm 0.00$   & $1.60 \pm 0.00$   & $0.90 \pm 0.00$   \\ \bottomrule
\end{tabular}
}
\end{table}

Training evaluates the velocity network $1+K_{-}+4T$ times per positive pair, and inference scores with Eq.~\eqref{eq:pred} at $O(d)$ per pair and $O(|\mathcal{I}|d)$ per user, the operations of MF. Tab.~\ref{tab:efficiency} measures every method with its selected configuration, with latency covering top-10 ranking over all items.

The costs follow from how each method scales with the catalog. The field of PACF has $O(d_h(d+d_h))$ parameters, independent of $|\mathcal{U}|$ and $|\mathcal{I}|$, and its cost per positive pair is fixed by $K_{-}$ and $T$. PACF accordingly has the lowest mean epoch time among the eight methods other than BPR-MF on Sports and Clothing, the two largest datasets, and its peak training memory stays near one gigabyte on every dataset. MultiVAE, RecVAE, DiffRec, and FlowCF process an interaction vector of length $|\mathcal{I}|$ per user, so their memory grows with the catalog and exceeds that of PACF on Sports and Clothing. Epoch times on the Amazon datasets vary widely across epochs, so the comparison uses means. At inference, PACF keeps the latency of an inner-product scorer, like BPR-MF and LightGCN. NCF and NeuMF run an MLP on every user-item pair, the variational autoencoders encode and decode each interaction vector, and DiffRec and FlowCF generate scores iteratively, DiffRec most slowly.

\subsection{Selected Hyperparameters}\label{sec:sensitivity}

\begin{table}[!t]
\centering
\caption{Selected PACF hyperparameters on ML, Food, Dance, Baby, Sports, and Clothing. Every dataset selects positive flow, trajectory, and negative-target weights and at most four RK4 steps, so the validation search keeps all pathwise terms of PACF active with short RK4 integration.}
\label{tab:best_hpo}
\begin{tabular}{lccccccc}
\toprule
Dataset  & $d$ & $d_h$ & $T$ & $\eta$              & $\lambda$ & $\beta$ & $\mu_-$ \\ \midrule
ML       & 256 & 128 & 1 & $4.83\times 10^{-2}$  & 0.86      & 0.25    & 0.71    \\
Food     & 512 & 64  & 4 & $2.26\times 10^{-2}$  & 0.41      & 0.18    & 0.34    \\
Dance    & 512 & 512 & 1 & $1.40\times 10^{-2}$  & 0.97      & 0.17    & 0.23    \\
Baby     & 512 & 64  & 4 & $1.18\times 10^{-2}$  & 0.21      & 0.44    & 0.76    \\
Sports   & 256 & 512 & 1 & $3.29\times 10^{-4}$  & 0.98      & 0.30    & 0.99    \\
Clothing & 512 & 256 & 4 & $9.66\times 10^{-5}$  & 0.69      & 0.11    & 0.05    \\ \bottomrule
\end{tabular}
\end{table}

Tab.~\ref{tab:best_hpo} lists the configuration that validation selects for each dataset, namely the embedding dimension $d$, hidden width $d_h$, RK4 steps $T$, learning rate $\eta$, flow weight $\lambda$, trajectory weight $\beta$, and negative-target scale $\mu_-$. The search ranges of $\beta$ and $\mu_-$ include zero, and every listed dataset selects positive values for both, so validation keeps the trajectory term and the signed negative regression. The flow weight has a positive lower bound by design. Through Eq.~\eqref{eq:sign_factor}, $\mu_-$ sets both the size of the negative target and the weight $\mu_-^2$ of the reparameterized negative regression, and its selected value varies widely, down to a small value on Clothing. Every dataset selects at most four RK4 steps, in line with the integration-depth analysis, and the two largest datasets select smaller learning rates.

\section{Theoretical Analysis}\label{sec:extended_theory}\label{app:proofs}

\subsection{Straight Reference Path}\label{sec:ot_path_proof}

The reference displacement of a pair $(u,i)$ is $\mathbf{q}_i-\mathbf{p}_u$, and the segment in Eq.~\eqref{eq:path} traverses it at the constant velocity $\dot{\mathbf{h}}_t=\mathbf{q}_i-\mathbf{p}_u$~\citep{liu2023flow,park2024constant}. Among absolutely continuous paths $\mathbf{h}:[0,1]\to\mathcal{M}$ with $\mathbf{h}_0=\mathbf{p}_u$ and $\mathbf{h}_1=\mathbf{q}_i$, this segment uniquely minimizes the kinetic energy $\mathcal{C}[\mathbf{h}]=\int_0^1\|\dot{\mathbf{h}}_t\|_2^2\,dt$ in Euclidean space~\citep{villani2008optimal,finlay2020train}. Its Euclidean interpolation follows Rectified Flow~\citep{liu2023flow} and flow matching~\citep{lipman2022flow} in embedding space. The endpoints give $\mathbf{q}_i-\mathbf{p}_u=\int_0^1\dot{\mathbf{h}}_t\,dt$~\citep{geng2025mean}. Jensen's inequality for the strictly convex $\|\cdot\|_2^2$ then gives $\|\mathbf{q}_i-\mathbf{p}_u\|_2^2\leq\mathcal{C}[\mathbf{h}]$, with equality exactly when $\dot{\mathbf{h}}_t$ is constant. The boundary conditions fix this constant at $\mathbf{q}_i-\mathbf{p}_u$, and integration from $\mathbf{p}_u$ recovers Eq.~\eqref{eq:path}.

For $t<1$, Eq.~\eqref{eq:path} gives $\mathbf{p}_u=(\mathbf{h}_t-t\mathbf{q}_i)/(1-t)$, so the signed target is the function $\mathbf{f}^{\star}$ of the complete network input in Eq.~\eqref{eq:target_function}. Under the positive label, $\mathbf{f}^{\star}$ returns $\mathbf{q}_i-\mathbf{p}_u$ at every state of the segment, so Eq.~\eqref{eq:node} with $\mathbf{f}^{\star}$ in place of $\mathbf{v}_{\boldsymbol\theta}$ has the solution in Eq.~\eqref{eq:path} for $0\leq t<1$, which converges to $\mathbf{q}_i$ as $t\to1$. The negative label reverses this velocity and scales it by $\mu_-$, so Eq.~\eqref{eq:target_vector} keeps the reference displacement of every pair and lets the label set its direction and magnitude. The last RK4 stage evaluates the field at $t=1$, and App.~\ref{sec:item_end_proof} derives the constraint it places on the users of an item.

\subsection{Positive-Score Monotonicity and Ranking Alignment}\label{sec:motivation_analysis}\label{sec:alignment_ranking}

Along the positive path $\mathbf{h}^{(u,i^+)}_t=(1-t)\mathbf{p}_u+t\mathbf{q}_{i^+}$, define the positive score $\psi(t)\triangleq\langle\mathbf{h}^{(u,i^+)}_t,\mathbf{q}_{i^+}\rangle$, which is the score $s(u,i^+)$ of Eq.~\eqref{eq:pred} with the user embedding moved to the path state.

\begin{proposition}[Pairwise Ranking Alignment]
\label{prop:pairwise_alignment}
Fix $\mathbf{q}_{i^+}$ and $\mathbf{q}_{i^-}$, and let $\ell_{\mathrm{BPR}}(t)=-\log\sigma\bigl(\langle\mathbf{h}^{(u,i^+)}_t,\mathbf{q}_{i^+}-\mathbf{q}_{i^-}\rangle\bigr)$. Then $\ell_{\mathrm{BPR}}$ is non-increasing in $t$ exactly when this condition holds:
\begin{equation}
\label{eq:pairwise_condition}
\langle\mathbf{q}_{i^+}-\mathbf{p}_u,\,\mathbf{q}_{i^+}-\mathbf{q}_{i^-}\rangle\geq0.
\end{equation}
\end{proposition}

\begin{proof}
Eq.~\eqref{eq:path} makes $\psi$ and the margin $\Delta(t)=\langle\mathbf{h}^{(u,i^+)}_t,\mathbf{q}_{i^+}-\mathbf{q}_{i^-}\rangle$ linear in $t$, with slopes:
\begin{equation}
\label{eq:phi_derivative}
\psi'(t)=\|\mathbf{q}_{i^+}\|_2^2-\langle\mathbf{p}_u,\mathbf{q}_{i^+}\rangle,
\qquad
\Delta'(t)=\psi'(t)-\langle\mathbf{q}_{i^+}-\mathbf{p}_u,\mathbf{q}_{i^-}\rangle.
\end{equation}
As $-\log\sigma$ strictly decreases, $\ell_{\mathrm{BPR}}$ is non-increasing exactly when $\Delta'\geq0$, which is Eq.~\eqref{eq:pairwise_condition}.
\end{proof}

By Eq.~\eqref{eq:phi_derivative}, the positive score is non-decreasing in $t$ if and only if the user embedding satisfies:
\begin{equation}
\label{eq:monotonicity_condition}
\langle\mathbf{p}_u,\mathbf{q}_{i^+}\rangle\leq\|\mathbf{q}_{i^+}\|_2^2,
\end{equation}
and strictly increasing under strict inequality. Equal norms $\|\mathbf{p}_u\|_2=\|\mathbf{q}_{i^+}\|_2$ satisfy this condition because $\psi'(t)$ then equals $\tfrac12\|\mathbf{q}_{i^+}-\mathbf{p}_u\|_2^2$. For a differentiable loss with $\partial\ell/\partial s(u,i^+)<0$ and fixed competing scores, the chain rule $d\ell/dt=(\partial\ell/\partial s(u,i^+))\,\psi'(t)$ makes $\ell$ non-increasing under Eq.~\eqref{eq:monotonicity_condition}, which covers BPR, binary cross-entropy on the positive pair, and sampled softmax. The margin slope in Eq.~\eqref{eq:phi_derivative} shows that Eq.~\eqref{eq:monotonicity_condition} implies Eq.~\eqref{eq:pairwise_condition} whenever $\langle\mathbf{q}_{i^+}-\mathbf{p}_u,\mathbf{q}_{i^-}\rangle\leq0$.

\textbf{Measured condition.}
On retuned checkpoints, adding the pathwise losses reduces the violation rate $\Pr[\psi'<0]$ of the positive-score slope in Eq.~\eqref{eq:phi_derivative} from $0.044804$ to $0.000074$ on ML and from $0.002392$ to $0.001233$ on Food, so almost every measured pair satisfies Eq.~\eqref{eq:monotonicity_condition}. For these pairs, moving the user toward the positive item with competing scores fixed lowers every loss that decreases in the positive score, and the pairwise margin follows the alignment condition of Eq.~\eqref{eq:pairwise_condition}.

\subsection{Conditional Flow Regularity}\label{sec:conditional_flow_regularity}

The state space is $\mathcal{M}=\mathbb{R}^d$. The exact-flow results below assume that for every fixed item and label $(\mathbf{q}_i,y_{ui})$ the field $\mathbf{v}_{\boldsymbol\theta}$ is continuously differentiable in $\mathbf{h}\in\mathcal{M}$, continuous in $t\in[0,1]$, and $L_{\mathbf h}$-Lipschitz in $\mathbf{h}$ with one state sensitivity $L_{\mathbf h}$ over all states, times, and item-label contexts.

We use $\varphi_{s\to t}$ for the exact map from time $s$ to time $t$ and write $\varphi_t=\varphi_{0\to t}$. The corresponding finite-step RK4 map is $\Phi_{s\to t}$, with $\Phi_1=\Phi_{0\to1}$ in the training objective. For a fixed item and label, the exact trajectory satisfies the conditional initial-value problem over the complete interval:
\begin{equation}
\label{eq:conditional_node_full}
\dot{\widetilde{\mathbf{h}}}^{(u,i)}_t
=\mathbf{v}_{\boldsymbol\theta}(\widetilde{\mathbf{h}}^{(u,i)}_t,t,\mathbf{q}_i,y_{ui}),
\qquad
\widetilde{\mathbf{h}}^{(u,i)}_0=\mathbf{p}_u.
\end{equation}
The exact map from the initial state $\mathbf{p}_u$ to time $t$ has the following equivalent integral form:
\begin{equation}
\label{eq:flow_map_def_cond}
\varphi_t(\mathbf{p}_u,\mathbf{q}_i,y_{ui})
\triangleq
\mathbf{p}_u
+
\int_{0}^{t}
\mathbf{v}_{\boldsymbol\theta}\!\left(\widetilde{\mathbf{h}}^{(u,i)}_\tau,\, \tau,\, \mathbf{q}_i,\, y_{ui}\right)\,d\tau.
\end{equation}

\begin{remark}[Exact Conditional Flow]
\label{rem:conditional_flow}
Under this regularity condition, the Picard--Lindel\"of theorem gives Eq.~\eqref{eq:conditional_node_full} a unique solution on $[0,1]$ for every $\mathbf{p}_u\in\mathcal{M}$, and the same argument backward in time makes $\varphi_t(\cdot,\mathbf{q}_i,y_{ui})$ a bijection with inverse $\varphi_{t\to0}$. Its Jacobian $\mathbf{J}(t)=\partial\varphi_t/\partial\mathbf{p}$ solves $\dot{\mathbf{J}}(t)=(\partial\mathbf{v}_{\boldsymbol\theta}/\partial\mathbf{h})\mathbf{J}(t)$ with $\mathbf{J}(0)=\mathbf{I}$, and Liouville's formula $\det\mathbf{J}(t)=\exp\int_0^t\tr(\partial\mathbf{v}_{\boldsymbol\theta}/\partial\mathbf{h})\,d\tau>0$ keeps it invertible, so $\varphi_t$ is a global $C^1$-diffeomorphism of the state space $\mathcal{M}$ for every $t\in[0,1]$.
\end{remark}

For two users with the same item and label, Gr\"{o}nwall's inequality applied to Eq.~\eqref{eq:flow_map_def_cond} gives $\|\varphi_1(\mathbf{p}_u,\mathbf{q}_i,y)-\varphi_1(\mathbf{p}_v,\mathbf{q}_i,y)\|_2\leq e^{L_{\mathbf h}}\|\mathbf{p}_u-\mathbf{p}_v\|_2$. A negative one-sided Lipschitz constant turns this bound into a contraction that pulls the exact endpoints of the two users together.

\begin{figure}[!t]
\centering
\begin{minipage}[t]{0.49\linewidth}
  \centering
  \includegraphics[width=\linewidth]{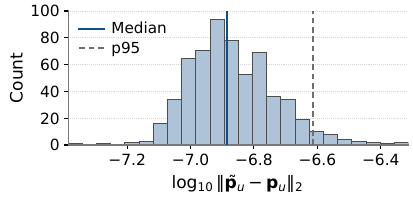}
  \caption{Distribution of RK4 round-trip residuals over ML users at one PACF checkpoint. The solid and dashed lines mark the median and the 95th percentile. The residuals stay below $10^{-6}$ for every user, so the RK4 map that PACF trains is nearly invertible for a fixed item and label.}
  \label{fig:diffeomorphism_errors}
\end{minipage}\hfill
\begin{minipage}[t]{0.49\linewidth}
  \centering
  \includegraphics[width=0.69\linewidth]{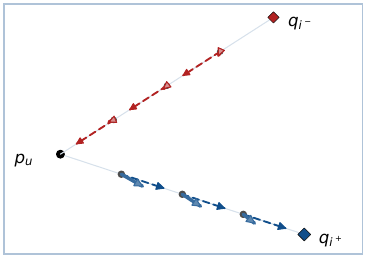}
  \caption{Learned field and BCF target for positive and negative ML pairs at a PACF checkpoint. Blue and red mark positive and negative pairs, dashed and solid arrows the target $\mathbf{u}_t$ and the field $\mathbf{v}_{\boldsymbol\theta}$, dots the sampled states, and diamonds the items. The field closely follows both signed targets, and the small remaining gaps are the residuals that update PACF's user embeddings.}
  \label{fig:bcf_alignment}
\end{minipage}
\end{figure}

Fig.~\ref{fig:diffeomorphism_errors} integrates each ML user embedding forward and back with eight RK4 steps, $\tilde{\mathbf{p}}_u=(\Phi_{1\to 0}\circ\Phi_{0\to 1})(\mathbf{p}_u)$, and plots the distribution of $\log_{10}\|\tilde{\mathbf{p}}_u-\mathbf{p}_u\|_2$ over users at one PACF checkpoint. The round trip returns every user close to its starting point for a fixed item and label, so the finite-step map $\Phi_1$ that PACF trains inherits the invertibility that Remark~\ref{rem:conditional_flow} establishes for the exact flow.

Because $\mathcal{L}_{\mathrm{traj}}$ averages $\|\Phi_1(\mathbf{p}_u)-\mathbf{q}_{i^+}\|_2^2$ over training positives, Jensen's inequality gives $\mathbb{E}_{(u,i^+)}\|\Phi_1(\mathbf{p}_u)-\mathbf{q}_{i^+}\|_2\leq\sqrt{\mathcal{L}_{\mathrm{traj}}}$ for the mean RK4 endpoint distance. Candidate ranking uses the static score in Eq.~\eqref{eq:pred}, which BPR trains directly and both pathwise losses shape.

\subsection{Shared-Field Regression and Local Consistency}\label{sec:segment_crossing}

PACF fits the signed targets of all pairs with one network $\mathbf{v}_{\boldsymbol\theta}$, so an update driven by one pair changes the predictions and residuals at the inputs of other pairs, including other users of the same item.

Let $\mathcal{D}$ be the distribution of training tuples $(u,i,y_{ui},t)$ from Alg.~\ref{alg:training}, with $t\sim\Unif[0,1]$ independent of $(u,i,y_{ui})$. Under $\mathcal{D}$, the state $\mathbf{h}_t=(1-t)\mathbf{p}_u+t\mathbf{q}_i$, target $\mathbf{u}_t=\mu(y_{ui})(\mathbf{q}_i-\mathbf{p}_u)$, and complete input $\mathbf{X}\triangleq(\mathbf{h}_t,t,\mathbf{q}_i,y_{ui})$ are random vectors. All conditional expectations use $\supp(\mathbf{X})$. Letting the network range over all measurable fields $\mathbf{v}:\supp(\mathbf{X})\to\mathcal{M}$, Eq.~\eqref{eq:flow_loss} becomes the population risk:
\begin{equation}
\label{eq:flow_population}
\mathcal{L}_{\mathrm{flow}}(\mathbf{v}) = \mathbb{E}_{\mathcal{D}}\!\left[\,\bigl\|\mathbf{v}(\mathbf{X}) - \mathbf{u}_t\bigr\|_2^2\,\right].
\end{equation}

Over all measurable fields, Eq.~\eqref{eq:flow_population} is minimized by the conditional mean $\mathbf{v}^*(\mathbf{x})=\mathbb{E}_{\mathcal{D}}[\mathbf{u}_t\mid\mathbf{X}=\mathbf{x}]=\mu(y)\bigl(\mathbf{q}-\mathbb{E}_{\mathcal{D}}[\mathbf{p}_u\mid\mathbf{X}=\mathbf{x}]\bigr)$ for $\mathbf{X}$-almost every $\mathbf{x}=(\mathbf{h},t,\mathbf{q},y)$, and for $t<1$ Eq.~\eqref{eq:path} pins $\mathbf{p}_u$ to $\mathbf{p}=(\mathbf{h}-t\mathbf{q})/(1-t)$, giving $\mathbf{v}^*(\mathbf{x})=\mu(y)(\mathbf{q}-\mathbf{p})$. For $t<1$, the complete input therefore fixes the target function in Eq.~\eqref{eq:target_function}, and the population optimum matches each target pointwise.

\textbf{Signed regression.}
For $\mu_->0$ and either label, every field output $\mathbf{v}$ satisfies the identity:
\begin{equation}
\label{eq:sign_factor}
\bigl\|\mathbf{v}-\mu(y)(\mathbf{q}_i-\mathbf{p}_u)\bigr\|_2^2=\mu(y)^2\,\bigl\|\mathbf{v}/\mu(y)-(\mathbf{q}_i-\mathbf{p}_u)\bigr\|_2^2.
\end{equation}
Under the output reparameterization $\mathbf{g}=\mathbf{v}/\mu(y)$, each flow term regresses $\mathbf{g}$ onto the unsigned displacement with weight $\mu(y)^2$, which is $1$ for positive and $\mu_-^2$ for negative pairs. The label thus sets the target direction, and $\mu_-$ acts as a quadratic weight of this reparameterized regression.

Cross-pair coupling requires a bound on how much the shared prediction varies. If $\mathbf{v}$ is $L_{\mathbf X}$-Lipschitz in $(\mathbf{h}_t,t,\mathbf{q}_i)$ at each fixed label and $B_{y_0}(\mathbf{x}_0,r)$ is the set of label-$y_0$ inputs whose continuous components lie within distance $r$ of those of $\mathbf{x}_0=(\mathbf{h}_0,t_0,\mathbf{q}_0,y_0)\in\supp(\mathbf{X})$, with $\Pr_{\mathcal{D}}(\mathbf{X}\in B_{y_0}(\mathbf{x}_0,r))>0$, then Jensen's inequality gives $\|\mathbf{v}(\mathbf{x}_0)-\mathbb{E}_{\mathcal{D}}[\mathbf{v}(\mathbf{X})\mid\mathbf{X}\in B_{y_0}(\mathbf{x}_0,r)]\|_2\leq L_{\mathbf X}r$.

This bound relates training pairs whose states, times, and item conditions lie within $r$ under one label, which is stricter than a crossing of two segments. Within such a neighborhood the shared field varies by at most $L_{\mathbf X}r$, while each pair compares the prediction with its own signed target, so their target difference enters their residuals and, through Eq.~\eqref{eq:flow_grad_residual}, their embedding gradients. Users of one item share its condition and label, and their states converge as $t\to1$, which yields Prop.~\ref{prop:item_end}.

Fig.~\ref{fig:bcf_alignment} projects one positive and one negative ML pair onto two principal components with a shared scale and aligned user endpoints. It shows the constant signed targets and the state-dependent predictions at three sampled states per segment, and the gaps between them are the fitting residuals that Eq.~\eqref{eq:flow_grad_residual} maps through the input and target derivatives into the user embedding $\mathbf{p}_u$.

\subsection{Proofs of the Item-End Results}\label{sec:item_end_proof}

The proofs use Eqs.~\eqref{eq:path} and~\eqref{eq:target_vector}. Prop.~\ref{prop:item_end} uses a state-Lipschitz condition, implied by the regularity condition in App.~\ref{sec:conditional_flow_regularity}, and Remark~\ref{rem:traj_endpoint} uses the RK4 step in Eq.~\eqref{eq:rk4_one_step} with step size $1/T$.

\begin{proof}[Proof of Proposition~\ref{prop:item_end}]
\textbf{Step 1. Pointwise bound.}
Write $\mathbf{h}_u=\mathbf{h}^{(u,i)}_t$, $\mathbf{h}_v=\mathbf{h}^{(v,i)}_t$, and $\mathbf{v}_{\boldsymbol\theta}(\mathbf{h})=\mathbf{v}_{\boldsymbol\theta}(\mathbf{h},t,\mathbf{q}_i,y)$. States differ by $(1-t)(\mathbf{p}_u-\mathbf{p}_v)$, and the targets differ by $-\mu(y)(\mathbf{p}_u-\mathbf{p}_v)$, so the target difference splits into the difference of field outputs and the difference of residuals:
\begin{equation}
-\mu(y)(\mathbf{p}_u-\mathbf{p}_v)
=\bigl[\mathbf{v}_{\boldsymbol\theta}(\mathbf{h}_u)-\mathbf{v}_{\boldsymbol\theta}(\mathbf{h}_v)\bigr]-\bigl[\mathbf{e}_u(t)-\mathbf{e}_v(t)\bigr].
\end{equation}
The triangle inequality and the Lipschitz condition on the segment between $\mathbf{h}_u$ and $\mathbf{h}_v$ give $|\mu(y)|\,\|\mathbf{p}_u-\mathbf{p}_v\|_2\leq L_{\mathbf h}(1-t)\,\|\mathbf{p}_u-\mathbf{p}_v\|_2+\|\mathbf{e}_u(t)\|_2+\|\mathbf{e}_v(t)\|_2$. Rearranging gives Eq.~\eqref{eq:item_end_pointwise} when $|\mu(y)|>L_{\mathbf h}(1-t)$, and otherwise Eq.~\eqref{eq:item_end_pointwise} holds trivially because its right side is zero.

\textbf{Step 2. Integration over $t$.}
Since $(a+b)^2\leq2(a^2+b^2)$, Eq.~\eqref{eq:item_end_pointwise} gives the squared bound:
\begin{equation}
\|\mathbf{e}_u(t)\|_2^2+\|\mathbf{e}_v(t)\|_2^2
\geq
\tfrac{1}{2}\bigl(|\mu(y)|-L_{\mathbf h}(1-t)\bigr)_+^2\,\|\mathbf{p}_u-\mathbf{p}_v\|_2^2.
\end{equation}
With $s=1-t$, the positive part vanishes for $s\geq|\mu(y)|/L_{\mathbf h}$, which is at most $1$ because $L_{\mathbf h}\geq|\mu(y)|$, so the expected squared coefficient over $t\sim\Unif[0,1]$ has the following closed form:
\begin{equation}
\mathbb{E}_{t\sim\Unif[0,1]}\bigl(|\mu(y)|-L_{\mathbf h}(1-t)\bigr)_+^2
=\int_0^{|\mu(y)|/L_{\mathbf h}}\bigl(|\mu(y)|-L_{\mathbf h}s\bigr)^2\,ds
=\frac{|\mu(y)|^3}{3L_{\mathbf h}}.
\end{equation}
Taking the expectation of the previous display therefore gives the following pairwise bound:
\begin{equation}
\label{eq:item_end_pair}
\mathbb{E}_t\bigl[\|\mathbf{e}_u(t)\|_2^2+\|\mathbf{e}_v(t)\|_2^2\bigr]
\geq
\frac{|\mu(y)|^3}{6L_{\mathbf h}}\,\|\mathbf{p}_u-\mathbf{p}_v\|_2^2.
\end{equation}

\textbf{Step 3. Dispersion form.}
Let $n=|\mathcal{S}|$. For $n=1$, Eq.~\eqref{eq:item_end_dispersion} has a zero right side. For $n\geq2$, sum Eq.~\eqref{eq:item_end_pair} over unordered pairs. Each user appears $n-1$ times. The identity $\sum_{\{u,v\}\subseteq\mathcal{S}}\|\mathbf{p}_u-\mathbf{p}_v\|_2^2=n\sum_{u\in\mathcal{S}}\|\mathbf{p}_u-\bar{\mathbf{p}}_{\mathcal{S}}\|_2^2$ converts pairwise distances to dispersion, giving the summed bound:
\begin{equation}
(n-1)\,\mathbb{E}_t\sum_{u\in\mathcal{S}}\|\mathbf{e}_u(t)\|_2^2
\geq
\frac{|\mu(y)|^3}{6L_{\mathbf h}}\,n\sum_{u\in\mathcal{S}}\|\mathbf{p}_u-\bar{\mathbf{p}}_{\mathcal{S}}\|_2^2.
\end{equation}
Dividing both sides by $n-1$ and using $n/(n-1)\geq1$ yields Eq.~\eqref{eq:item_end_dispersion} for every $n\geq2$.
\end{proof}

\begin{remark}[Shared Final RK4 Stage]
\label{rem:traj_endpoint}
Fix item $i$, a set $\mathcal{S}$ of users observed with it, and $\mathbf{r}_u=\Phi_1(\mathbf{p}_u\mid\mathbf{q}_i,+1)-\mathbf{q}_i$. If the positive field returns $\mathbf{q}_i-\mathbf{p}_u$ at the reference states of every $u\in\mathcal{S}$ at the stage times $n/T$ and $(n+\frac12)/T$ for $n<T$, the last RK4 stage evaluates the same input $(\mathbf{q}_i,1,\mathbf{q}_i,+1)$ for all these users. With $\mathbf{c}_i$ the output there, $\mathbf{r}_u=(\mathbf{c}_i-\mathbf{q}_i+\mathbf{p}_u)/(6T)$, so every output $\mathbf{c}_i$ gives $\sum_{u\in\mathcal{S}}\|\mathbf{r}_u\|_2^2\geq(6T)^{-2}\sum_{u\in\mathcal{S}}\|\mathbf{p}_u-\bar{\mathbf{p}}_{\mathcal{S}}\|_2^2$, with equality if and only if $\mathbf{c}_i=\mathbf{q}_i-\bar{\mathbf{p}}_{\mathcal{S}}$.
\end{remark}

\textbf{Shared final stage.} This paragraph proves Remark~\ref{rem:traj_endpoint}. Write $t_n=n/T$, and let step $n$ apply Eq.~\eqref{eq:rk4_one_step} on $[t_n,t_{n+1}]$ with step size $1/T$. If $\Phi_{0\to t_n}(\mathbf{p}_u)=\mathbf{h}^{(u,i)}_{t_n}$, which holds at $n=0$, the stages of step $n$ evaluate reference states at $t_n$, at $t_n+\tfrac{1}{2T}$ twice, and at $t_{n+1}$. For $n<T-1$ all four return $\mathbf{q}_i-\mathbf{p}_u$, so $\Phi_{0\to t_{n+1}}(\mathbf{p}_u)=\mathbf{h}^{(u,i)}_{t_{n+1}}$, and induction reaches $t_{T-1}$. In the last step the first three stages return $\mathbf{q}_i-\mathbf{p}_u$, while the fourth evaluates $\mathbf{h}^{(u,i)}_{t_{T-1}}+\tfrac{1}{T}(\mathbf{q}_i-\mathbf{p}_u)=\mathbf{q}_i$ at $t=1$ and returns $\mathbf{c}_i$. Hence $\Phi_1(\mathbf{p}_u)=\mathbf{q}_i+(\mathbf{c}_i-\mathbf{q}_i+\mathbf{p}_u)/(6T)$. Since $\sum_{u\in\mathcal{S}}\|\mathbf{p}_u-\mathbf{a}\|_2^2$ is smallest at $\mathbf{a}=\bar{\mathbf{p}}_{\mathcal{S}}$, the choice $\mathbf{a}=\mathbf{q}_i-\mathbf{c}_i$ gives the bound in Remark~\ref{rem:traj_endpoint}, with equality at $\mathbf{c}_i=\mathbf{q}_i-\bar{\mathbf{p}}_{\mathcal{S}}$. Under this fit one user reaches $\mathbf{q}_i$ exactly, while distinct users keep a positive total endpoint residual. Euler and midpoint steps evaluate the field only at $t<1$, so this coupling comes from the final RK4 stage.

\textbf{Item-end intuition.} If the field gives nearly the same prediction at the states of the users in $\mathcal{S}$, the flow loss is smallest when this prediction equals their mean target $\mu(y)(\mathbf{q}_i-\bar{\mathbf{p}}_{\mathcal{S}})$, and the residual of user $u$ is then $\mu(y)(\mathbf{p}_u-\bar{\mathbf{p}}_{\mathcal{S}})$. The target route of Eq.~\eqref{eq:flow_grad_residual} then contributes $2\mu(y)^2(\mathbf{p}_u-\bar{\mathbf{p}}_{\mathcal{S}})$ to the gradient, so a descent step moves $\mathbf{p}_u$ toward the group mean, while the input route carries the factor $1-t$ and vanishes at the item end. For $y=+1$ the group holds the observed users of $i$, and for $y=-1$ it holds the users whose sampled negatives include $i$, with weight $\mu_-^2\leq1$.

\textbf{State sensitivity.}
For fixed $(t,\mathbf{q},y)$ with $t<1$, $\mathbf{f}^{\star}$ in Eq.~\eqref{eq:target_function} has state Jacobian $-\mu(y)\mathbf{I}/(1-t)$ and state sensitivity $|\mu(y)|/(1-t)$, which exceeds a fixed $L_{\mathbf h}$ exactly when $1-t<|\mu(y)|/L_{\mathbf h}$ and grows without bound as $t\to1$ for $\mu(y)\neq0$. A field matching this Jacobian near the segments requires $L_{\mathbf h}\geq|\mu(y)|$. If $0<L_{\mathbf h}<|\mu(y)|$, integration instead gives the positive coefficient $\bigl(|\mu(y)|^3-(|\mu(y)|-L_{\mathbf h})^3\bigr)/(6L_{\mathbf h})$ multiplying the user dispersion for that item and label.

\subsection{Embedding Gradient Derivations}
\label{sec:gradient_analysis}

\begin{proposition}[Residual-Dependent Flow Gradients]
\label{prop:flow_residual_gradients}
Fix a pair $(u,i)$ with label $y=y_{ui}$ and a time $t\in[0,1)$. Under Eq.~\eqref{eq:path}, the BCF target of Eq.~\eqref{eq:target_vector} is the following function of the network input:
\begin{equation}
\label{eq:target_function}
\mathbf{f}^{\star}(\mathbf{h},t,\mathbf{q},y)=\mu(y)\,\frac{\mathbf{q}-\mathbf{h}}{1-t}.
\end{equation}
Suppose the field is differentiable in its continuous inputs. Let $\mathbf{e}=\mathbf{v}_{\boldsymbol\theta}(\mathbf{h}_t,t,\mathbf{q}_i,y)-\mu(y)(\mathbf{q}_i-\mathbf{p}_u)$ and $\mathbf{J}_{\mathbf h}=\partial\mathbf{v}_{\boldsymbol\theta}/\partial\mathbf{h}$ at this input. With $\boldsymbol\theta$ and $\mathbf{q}_i$ fixed, the squared residual has the user gradient:
\begin{equation}
\label{eq:flow_grad_residual}
\nabla_{\mathbf{p}_u}\|\mathbf{e}\|_2^2=2\bigl[(1-t)\mathbf{J}_{\mathbf h}+\mu(y)\mathbf{I}\bigr]^{\top}\mathbf{e}.
\end{equation}
\end{proposition}
Both gradients are linear in the residual, so an exactly fitted pair contributes no flow gradient. The user gradient combines the state route $(1-t)\mathbf{J}_{\mathbf h}$ with the target route $\mu(y)\mathbf{I}$, and the item gradient also differentiates the item condition of the network. The label sets the target direction, and after the output rescaling in Eq.~\eqref{eq:sign_factor}, negative-pair regression carries weight $\mu_-^2$ relative to positive pairs.

\begin{proof}[Proof of Proposition~\ref{prop:flow_residual_gradients}]
Solving Eq.~\eqref{eq:path} gives $\mathbf{p}_u=(\mathbf{h}_t-t\mathbf{q}_i)/(1-t)$ and hence Eq.~\eqref{eq:target_function}. The state derivative with respect to $\mathbf{p}_u$ is $(1-t)\mathbf{I}$, while the target derivative is $-\mu(y)\mathbf{I}$. Applying the chain rule to $\|\mathbf{e}\|_2^2$ gives Eq.~\eqref{eq:flow_grad_residual}. With $\mathbf{J}_{\mathbf q}=\partial\mathbf{v}_{\boldsymbol\theta}/\partial\mathbf{q}$ at the same input, the item gradient is $\nabla_{\mathbf{q}_i}\|\mathbf{e}\|_2^2=2\bigl[t\mathbf{J}_{\mathbf h}+\mathbf{J}_{\mathbf q}-\mu(y)\mathbf{I}\bigr]^{\top}\mathbf{e}$, which adds the direct derivative through the item condition.
\end{proof}

\begin{proof}[Proof of Proposition~\ref{prop:user_gradient_decomposition}]
Eq.~\eqref{eq:pred} gives $\Delta s=\langle\mathbf{p}_u,\mathbf{q}_{i^+}-\mathbf{q}_{i^-}\rangle$. For the loss $\ell_{\mathrm{rec}}=-\log\sigma(\Delta s)$, the chain rule starts from the derivative with respect to the margin $\Delta s$, which takes the form:
\begin{equation}
\frac{\partial\ell_{\mathrm{rec}}}{\partial\Delta s}
=-\sigma(-\Delta s).
\end{equation}
Since $\partial\Delta s/\partial\mathbf{p}_u=\mathbf{q}_{i^+}-\mathbf{q}_{i^-}$, the chain rule gives the BPR contribution to the user update:
\begin{equation}
-\nabla_{\mathbf{p}_u}\ell_{\mathrm{rec}}
=\sigma(-\Delta s)(\mathbf{q}_{i^+}-\mathbf{q}_{i^-}).
\end{equation}

The flow contribution is $-\lambda\nabla_{\mathbf{p}_u}\ell_{\mathrm{flow}}$, where $\nabla_{\mathbf{p}_u}\ell_{\mathrm{flow}}$ sums Eq.~\eqref{eq:flow_grad_residual} over the flow terms of the triplet. For the trajectory contribution, use the endpoint residual $\mathbf{r}_u$ and integration Jacobian $\mathbf{J}_{\Phi}$ from Prop.~\ref{prop:user_gradient_decomposition}. Differentiating $\ell_{\mathrm{traj}}=\|\mathbf{r}_u\|_2^2$ with $\mathbf{q}_{i^+}$ and $\boldsymbol\theta$ fixed gives the trajectory contribution:
\begin{equation}
-\beta\nabla_{\mathbf{p}_u}\ell_{\mathrm{traj}}
=-2\beta\mathbf{J}_{\Phi}(\mathbf{p}_u)^\top\mathbf{r}_u.
\end{equation}
Adding the three negative gradients gives Eq.~\eqref{eq:grad_traj_refined}, the decomposition stated in Prop.~\ref{prop:user_gradient_decomposition}.
\end{proof}

\end{document}